\documentclass[mnsc,nonblindrev]{informs3}

\OneAndAHalfSpacedXI

\usepackage{natbib}
 \bibpunct[, ]{(}{)}{,}{a}{}{,}%
 \def\bibfont{\small}%
\setcitestyle{authoryear}

 \usepackage{enumitem}
 
 \TheoremsNumberedThrough 
\EquationsNumberedThrough

\usepackage{booktabs} 
\usepackage[ruled]{algorithm2e} 
\usepackage{subcaption}
\usepackage[framemethod=tikz]{mdframed}
\usepackage{subfiles}
\usepackage{pstricks}
\usepackage{bm}

\SetAlFnt{\small}
\SetAlCapFnt{\small}
\SetAlCapNameFnt{\small}
\SetAlCapHSkip{0pt}
\IncMargin{-\parindent}

\newcolumntype{M}[1]{>{\centering\arraybackslash}m{#1}}
\newcolumntype{N}{@{}m{0pt}@{}}

\newcommand{\I}[1]{\mathbb{I}\Big(#1\Big)}
\renewcommand{\P}[1]{\mathbb{P}\left(#1\right)}
\newcommand{\E}[1]{\mathbb{E}\left(#1\right)}
\newcommand{\set}[1]{\mathcal{#1}}
\newcommand{\don}{s}
\newcommand{\Don}{\mathcal{S}}
\newcommand{\V}{\mathcal{V}}
\newcommand{\D}{S}
\newcommand{\J}{J}
\newcommand{\rtwo}{\hat{r}}
\newcommand{\Jtwo}{\hat{J}}
\newcommand{\Jhat}{\tilde{J}}
\newcommand{\aone}{\alpha_1}
\newcommand{\atwo}{\alpha_2}
\newcommand{\y}{\tilde{x}}
\newcommand{\Z}{\zeta}
\newcommand{\cvar}{c}
\renewcommand{\h}{h}
\newcommand{\revcolor}{}
\newcommand{\sdnpcolor}{}
\newcommand{\FWstep}{{\revcolor m}}
\newcommand{\bigval}{n}
\newcommand{\bigvalvol}{n}

\begin{document}
\TITLE{Online Policies for Efficient Volunteer Crowdsourcing}
\MANUSCRIPTNO{MS-RMA-20-01950}


\RUNAUTHOR{Manshadi and Rodilitz}
	
\RUNTITLE{Efficient Volunteer Crowdsourcing}

\ARTICLEAUTHORS{%
\AUTHOR{Vahideh Manshadi}
\AFF{Yale School of Management, New Haven, CT, \EMAIL{vahideh.manshadi@yale.edu}}
\AUTHOR{Scott Rodilitz}
\AFF{Yale School of Management, New Haven, CT, \EMAIL{scott.rodilitz@yale.edu}}
} 

\ABSTRACT{
Nonprofit crowdsourcing platforms such as food recovery organizations  rely on volunteers to perform time-sensitive tasks. Thus, their success 
crucially depends on efficient volunteer utilization and engagement. {To encourage volunteers to complete a task, platforms use nudging mechanisms to notify a subset of volunteers with the hope that at least one of them responds positively. However, since excessive notifications may reduce volunteer engagement,
the platform faces a trade-off between notifying more volunteers for the current task and saving them for future ones.}  
Motivated by these applications, we introduce the online volunteer notification problem,  a generalization of online stochastic bipartite matching where tasks arrive following a known time-varying distribution over task types. Upon arrival of a task, the platform notifies a subset of volunteers with the objective of minimizing the number of missed tasks. To capture each volunteer's adverse reaction to excessive notifications, we assume that a notification triggers a random period of inactivity, during which she will ignore all notifications. However, if a volunteer is active and notified, she will perform the task with a given pair-specific {match} probability that captures her preference for the task. 
{\sdnpcolor We develop an online randomized policy that achieves a constant-factor guarantee close to the upper bound we establish for the performance of any online policy. Our policy as well as hardness results are parameterized by the minimum discrete hazard rate of the inter-activity time distribution.
The design of our policy relies on sparsifying an ex-ante feasible solution by solving a sequence of dynamic programs.}
Further, in collaboration with Food Rescue U.S., a volunteer-based food recovery platform, we demonstrate  the effectiveness of our policy by testing it on the platform's data from various locations across the U.S.
}


\maketitle

\section{Introduction}
\label{sec:intro}

Volunteers in the U.S. provide around $8$ billion hours of free labor annually. However, roughly $30\%$ of volunteers become disengaged  the following year, representing a loss of approximately $\$70$ billion in economic value as well as a significant challenge for the sustainability of organizations relying on volunteerism \citep{nationalservice2015, independentsector2018}.
Lack of retention partially stems from overutilization as well as the mismatch between a volunteer's preferences and the opportunities presented to her \citep{ locke2003hold, brudney2009ain}. 
The emergence of online volunteer crowdsourcing platforms presents a unique opportunity to 
design data-driven volunteer management tools that cater to volunteers' {heterogeneous} preferences. In the present work, we move toward this goal by taking {an algorithmic} approach to designing nudging mechanisms commonly used to encourage volunteers to perform tasks.

This work is motivated by our collaboration with {Food Rescue U.S. (FRUS), a nonprofit platform that recovers food from businesses and donates it to local agencies by crowdsourcing the transportation to volunteers.} 
In the following, we provide background on FRUS and highlight the challenge 
it faces when making volunteer nudging decisions. Further, we offer insights into volunteer behavior by analyzing FRUS data from different locations. Then, we list a summary of our contributions.

{\bf FRUS: A Crowdsourcing Platform for Food Recovery:} 
FRUS is a leading online platform that simultaneously addresses the societal problems of food waste and hunger. Over 60 million tons of food go to waste in the U.S. each year, while in 2018, 37 million people—including 11 million children—lived in food-insecure households \citep{food_waste, food_insecurity}. This mismatch is driven in part by the cost of last-mile transportation required to recover perishable donated food from local restaurants and grocery stores. 
FRUS has empowered donors by connecting them to local agencies and enabling free delivery through its dedicated volunteer base. Currently, it operates in tens of locations across different states, and so far it has recovered over 50 million pounds of food. On FRUS, a volunteering task---which is referred to as a {\em rescue}---involves transporting a prearranged, perishable food donation from a donor to  a local agency. Scheduled donations are often recurring and they are posted on the FRUS app in advance. While around $78\%$ of rescues are claimed organically by volunteers before the day of the rescue, around $22\%$ remain unclaimed on the last day.\footnote{Here, by organic, we mean volunteers sign up for those rescues without the platform's involvement.} 
In that case, to encourage volunteers to claim the rescue, FRUS {notifies} a subset of volunteers with the hope that at least one of them responds positively.
However, based on our interviews with the platform's local managers, FRUS faces a challenge when deciding whom to notify: on the one hand, it aims to minimize the probability of a missed rescue---which is achievable by notifying more volunteers.\footnote{Based on FRUS data, a missed rescue increases the probability of donor dropout by a factor of more than 2.5.} On the other hand, it wants to avoid excessive notifications because that may reduce volunteer engagement.\footnote{FRUS's current practice in {many} locations is to notify a volunteer at most once a week. Further, FRUS is hesitant to demand prompt responses from volunteers, which renders the option of sequentially notifying volunteers impractical. \label{footnote:currentpractice}}  

{Understanding volunteer behavior can help resolve the aforementioned trade-off}: if volunteers have preferences for certain rescues, then FRUS should mainly notify them for those tasks. Our analysis of two years of data indeed indicates that volunteer preferences are fairly consistent. To highlight this, in Figure \ref{fig:pcatotal} we visualize the first three principal components for characteristics of rescues completed by the most active volunteers in two FRUS locations. Each color represents a different volunteer, and the size of each circle is proportional to the frequency with which the volunteer completes a rescue of that type. 
For instance, more than 90\% of the rescues completed by the red volunteer in Location (a), as shown in Figure \ref{fig:PCA1}, are clustered within a cube whose volume is less than one tenth of the PCA component range.
As evident from these plots, volunteers tend to claim rescues that have similar characteristics, reflecting their geographical and time preferences.

Our interviews and empirical findings raise a key question that motivates our work: facing such volunteer behavior, how should a volunteer-based online platform, such as FRUS, design an effective {notification system} for  time-sensitive tasks? 

\begin{figure}
 \centering
 \begin{subfigure}[b]{0.48\textwidth}
 \includegraphics[width=\textwidth]{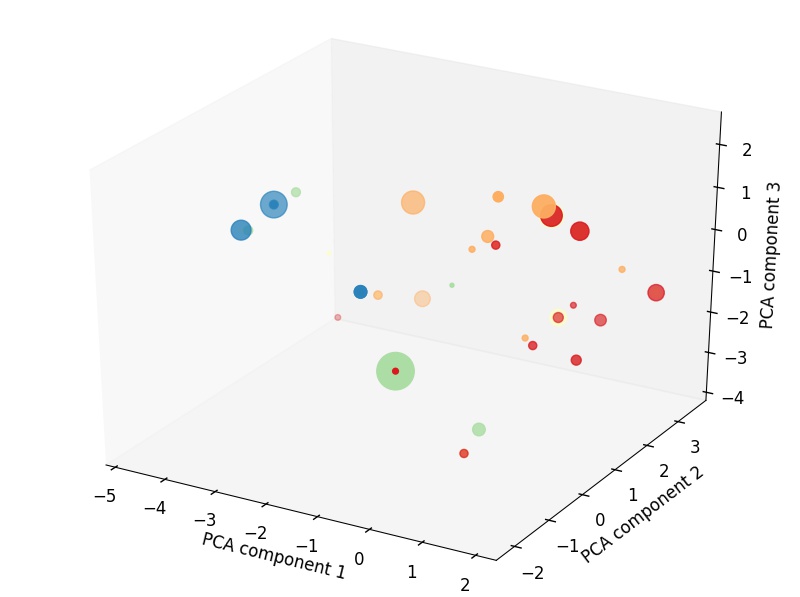}
  \caption{ }\label{fig:PCA1}
 \end{subfigure}
 \begin{subfigure}[b]{0.48\textwidth}
 \centering
 \includegraphics[width= \textwidth]{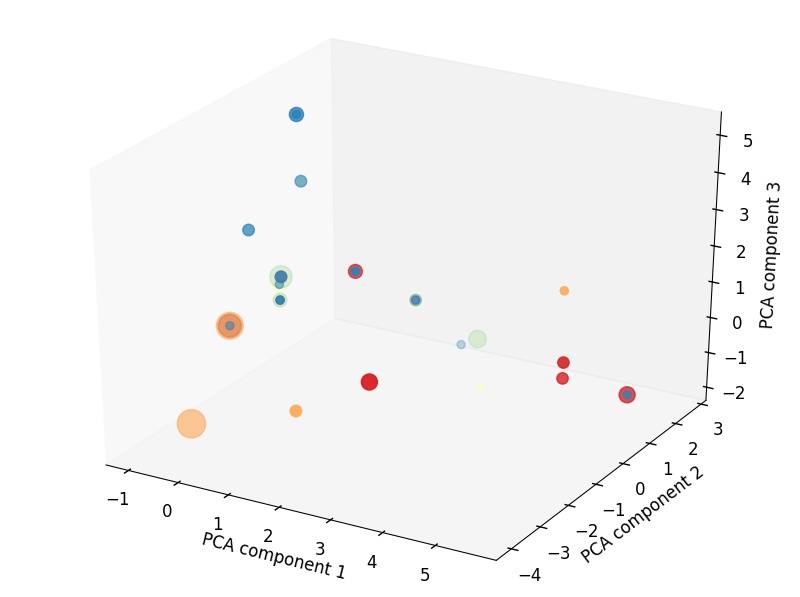}
 \caption{ }\label{fig:PCA2}
\end{subfigure}
\caption{Figure \ref{fig:PCA1} shows the first three principal components (PCs) for characteristics of rescues completed by the five most active volunteers
 in Location (a). Each color represents a different volunteer and the size of each circle is proportional to the frequency with which the volunteer completes a rescue with those PCs.
Figure \ref{fig:PCA2} shows the same plot for Location (b).}
\label{fig:pcatotal}
\end{figure}

{\bf Summary of Contributions:} Motivated by our collaboration with FRUS,
 we (i) introduce  the online volunteer notification problem which captures key features of volunteer  labor consistent with the literature, (ii) develop {\sdnpcolor an online randomized policy that achieves a constant-factor guarantee for the online volunteer notification problem,} (iii) establish an upper bound on the performance of any online policy, and (iv) demonstrate the effectiveness of our {\sdnpcolor policy by testing it} on FRUS's data from various locations across the U.S.

{\bf Modeling the Platform's Notification Problem:} We introduce the {\em online volunteer notification problem} to model a platform's notification decisions when utilizing volunteers to complete time-sensitive tasks. There are three main considerations that the platform should take into account: (i) volunteers' response to a notification is uncertain, (ii) the platform cannot expect volunteers to respond promptly, and (iii) if notified excessively, volunteers may suffer from notification fatigue. 
To include all of these considerations in our model, we assume that when each task arrives, the platform simultaneously notifies a {\em subset} of volunteers in the hope that at least one responds positively. To model a volunteer's adverse reaction toward excessive notifications, we assume that a volunteer can be in one of two possible states: {\em active} or {\em inactive}. In the former state, the volunteer pays attention to the platform’s notifications {and responds positively with her task-specific match probability, whereas in the latter state she ignores all notifications.} After being notified, an active volunteer will transition to the inactive state for a random inter-activity period. 
These three modeling considerations are aligned with the literature on volunteer management (see, e.g., \citealt{ata2016dynamic} and \citealt{dickerson2019blood}), where volunteers---unlike paid employees---cannot be \emph{required} to respond and typically go through periods of inactivity (as a consequence of utilization). Further, as mentioned in Footnote \ref{footnote:currentpractice}, in various locations, FRUS's managers follow the strategy of notifying a volunteer at most once per week and do not require prompt responses from volunteers. This practice fits within our modeling framework by setting the inter-activity period to be deterministic and equal to 7 days.

Because these platforms usually require the recurring completion of similar tasks, they can use historical data to predict their future last-minute needs. For instance, FRUS usually receives donations from the same source on a weekly basis. We model this by assuming that tasks belong to a given set of types and they arrive according to a (time-varying) distribution.
The platform makes online decisions aiming to maximize the number of completed tasks  knowing the arrival rates, match probabilities, and the inter-activity time distribution, but without observing the state of each volunteer.

Optimally resolving the trade-off between notifying a volunteer about the current task or saving her for the future depends jointly on the state of {\em all} volunteers. Consequently, determining the optimal notification strategy in the online volunteer notification problem is intractable due to the curse of dimensionality of the underlying dynamic program \emph{even when volunteers' states are known}. 
Thus, even for the special case of deterministic inter-activity periods the problem remains intractable.
In light of this challenge, we design {\sdnpcolor an online policy that can be computed in polynomial time,
and we prove a constant-factor guarantee on its performance. 
}

{\bf Developing an Online Policy:}  
{\sdnpcolor We develop a randomized policy based on an ex ante fractional solution that can be computed in polynomial time. 
In order to assess the performance of our policy,}
we use a {linear program benchmark} 
 whose optimal value serves as an upper bound on the value of a clairvoyant solution which knows the sequence of arrivals a priori as well as the state of volunteers at each time (see Program (LP),  Proposition \ref{prop:LP}, and Definition \ref{def:compratio}).
We remark that the platform's objective---{maximizing the number of completed tasks}---{\em jointly}  depends on the response of all volunteers and exhibits diminishing returns. For example, if the platform notifies two active volunteers $v$ and $u$ about a task of type $\don$, then the probability of completion would be $[1 - (1- p_{v,\don}) (1 - p_{u,\don})]$ where $p_{v,\don}$ and  $p_{u,\don}$ are the match probabilities of the pairs $(v,\don)$ and $(u,\don)$, respectively. {This} objective function presents two challenges: (i) an ex ante solution based on upper bounding such an objective  function by a piecewise linear one  {can be} ineffective in practice, and (ii) jointly analyzing volunteers' contribution for an online policy while keeping track of the joint distribution of their states (active or inactive) is prohibitively difficult, even in the special case of deterministic inter-activity times. We address the former by computing ex ante solutions that ``better'' approximate the true objective function as opposed to only relying on the LP solution (see Programs \eqref{MRG} and \eqref{SQ} and Proposition \ref{prop:fxstar}). We overcome the latter by assuming an artificial priority among volunteers which allows us to decouple their contributions (see Definition \ref{def:priority} and Lemma \ref{lem:falg}).

Attempting to follow the fractional ex ante solution can result in poor performance since, under such a policy, volunteers can become inactive at inopportune times (see Section \ref{subsec:upperboundexante} and Proposition \ref{prop:ubexante}). 
{\sdnpcolor Therefore, in the design of our policy, we modify the ex ante solution to account for inactivity while guaranteeing a constant-factor competitive ratio. Our policy,} the {\em sparse notification (SN) policy}, relies on {solving a sequence of dynamic programs (DPs)---one for each volunteer---to resolve the trade-off between notifying a volunteer now and saving her for future tasks. We solve the DPs in order of volunteers' artificial priorities, and each subsequent DP is formulated based on the previous solutions} (see Algorithm \ref{alg:two} and the preceding discussion).


{\sdnpcolor Our SN policy is parameterized by the minimum discrete hazard rate (MDHR) of the inter-activity time distribution, which serves as a sufficient condition for the level of  ``activeness'' of volunteers (see Definition \ref{def:hazard} and the following discussion).
We analyze the competitive ratio of our policy as a function of the MDHR. Our analysis relies on {decomposing the problem} into individual contributions based on our (artificial) priority scheme. We crucially use the dual-fitting framework of \citet{alaei2012online}, and our analysis relies on formulating a linear program along with its dual to place a lower bound on the optimal value of each volunteer's  DP (see Section \ref{subsec:alg2}).}\footnote{{\revcolor In Appendix \ref{app:sdnp}, we design and analyze a second policy, the {\em scaled-down notification (SDN) policy}, which achieves the same competitive ratio as the SN policy
(see Theorem \ref{thm:alg1}) 
using nearly identical computation time 
(see Remark \ref{remark:compcomplex}) 
by properly scaling down the notification probability prescribed by the ex-ante solution. However, the SN policy achieves significantly better performance than the SDN policy in the FRUS setting (see Figure \ref{fig:numericssdn}) and can perform nearly twice as well in certain instances (see the discussion in Appendix \ref{app:sdnp}). 
}}

{\bf {Upper Bound on Online Policies}:} In order to gain insight into the limitation of  online policies when compared to our benchmark, we develop  an upper bound on the achievable performance of any online policy, even policies which cannot be computed in polynomial time. {\sdnpcolor Like our policy, the upper bound is} parameterized by the MDHR (see Theorem \ref{thm:hardness}).
{As a consequence,} the gap between the achievable upper bound and our lower bound (attained through our {\sdnpcolor policy}) depends on the MDHR (see Figure \ref{fig:hardness}). When the MDHR is small, the gap is fairly small; however, the gap grows as the MDHR increases. 
Our upper bound relies on analyzing {\revcolor three instances} and is relatively tight when the MDHR is small.

{\bf Testing on FRUS Data:} 
In order to illustrate the effectiveness of our modeling approach and our {\sdnpcolor policy in practice, we  evaluate the performance of our SN policy by testing it} on FRUS's data from different locations. In Section \ref{sec:data}, we describe how we estimate model primitives and construct problem instances. Then we numerically show the superior performance of our {\sdnpcolor SN policy} when compared to different benchmarks, including strategies that resemble the current practice at different locations. Further, we present numerical results that demonstrate the robustness of our {\sdnpcolor policy} in the presence of small misspecifications of the model primitives, i.e., the arrival rates and the match probabilities (see  Appendix \ref{app:robustness} and Table \ref{table:robustness}).

While our collaboration with FRUS motivated us to introduce and study the online volunteer notification problem, our framework can be applied well beyond FRUS to a wide range of settings. Thousands of other nonprofits make use of platforms like DialMyCalls to send instantaneous notifications to their volunteer base.\footnote{Social network platforms such as Facebook have also been utilized as exemplified in the context of blood donation \citep{dickerson2019blood}.} Similar to FRUS, these nonprofits face the challenge of striking the right balance between notifying enough interested volunteers to fulfill an immediate need while avoiding excessive notification. Our framework and data-driven approach can be utilized in customizing such online notification systems. Moving beyond volunteer crowdsourcing, the negative impact of excessive notifications and marketing fatigue have been well-documented in marketing and social network engagement \citep{sinha2007over, cheng2010information, byers2012groupon}. In these applications, similar tensions arise between maximizing an immediate payoff (such as short-term engagement) and limiting notifications \citep{borgs2010game, lin2017overview, cao2019sequential}. As such, our general framework can also be applied to managing notification fatigue in the aforementioned contexts.

The rest of the paper is organized as follows. In Section \ref{sec:lit}, we review the related literature. In Section \ref{sec:model}, we formally introduce the online volunteer notification problem as well as the benchmark and the measure of competitive ratio.  
Section \ref{sec:algos} is the main algorithmic section of the paper and is devoted to describing and analyzing our {\sdnpcolor online policy}. 
In Section \ref{sec:hardness}, we present our upper bound on the achievable competitive ratio of 
any online policy as well as an upper bound on the performance of following the ex ante solution. 
In Section \ref{sec:data}, we revisit the FRUS application and {\sdnpcolor test our policy} on the platform’s data from various locations. Section \ref{sec:conclude} concludes the paper.

\section{Related Work}
\label{sec:lit}
Our work relates to and contributes to several streams of literature.

{\bf Volunteer Operations and Staffing}:
Due to the differences between volunteer and traditional labor as highlighted in \citet{sampson2006optimization}, managing a volunteer workforce provides unique challenges and opportunities that have been studied in the literature using various methodologies \citep{gordon2004improving, falasca2012optimization, lacetera2014rewarding,  sonmez2016improving, ata2016dynamic, dickerson2019blood, urrea2019volunteer, lo2021commitment, ata2021dynamic}. One key operational challenge is the uncertainty in both volunteer labor supply and demand. Using an elegant queuing model, \citet{ata2016dynamic} studies the problem of volunteer staffing with an application to gleaning organizations. Our approach to modeling volunteer behavior (specifically, assuming that notifying an active volunteer triggers a random period of inactivity) bears {some resemblance to} the modeling approach taken in \citet{ata2016dynamic}.

In a novel recent  work, \citet{dickerson2019blood} studies the problem of matching blood donors to donation centers, assuming that donors have preferences (over centers) and constraints on the frequency of receiving notifications. Using a stochastic matching policy, they demonstrate strong numerical performance relative to various benchmarks. 
{There are some similarities between our modeling approach and the approach used in \citet{dickerson2019blood}, but we highlight  {three} key differences.} (i) While their work focuses on the numerical evaluation of policies, we theoretically analyze the performance of our {\sdnpcolor policy} and provide an upper bound on the performance achievable by any online policy {(see Theorems \ref{thm:alg2} and \ref{thm:hardness})}.(ii) We model volunteers' adverse reactions to excessive notifications in a general form by considering arbitrary inter-activity time distributions. (iii) We parameterize our achievable upper and lower bounds by the minimum discrete hazard rate of that distribution.

{\bf  Crowdsourcing Platforms}:
Reflecting the growth of online technologies, there is a burgeoning literature on the operations of crowdsourcing platforms. 
Examples of {such }work include \citet{karger2014budget} for task crowdsourcing; \citet{ hu2015product} and  \citet{alaei2016dynamic} for crowdfunding;  \citet{Ilan} and \citet{Jacob} for crowdsourcing in transportation and urban mobility; and \citet{ali, chris, nikhil}, and \citet{papanastasiou2018crowdsourcing} for information crowdsourcing.
Our work adds to the growing collection of papers that focus specifically on nonprofit crowdsourcing platforms, with applications as varied as educational crowdfunding \citep{song2018matching}, disaster response \citep{han2019harnessing}, and smallholder supply chains \citep{de2019crowdsourcing}. {Nonprofits often {cannot rely on} monetary incentives; in such settings, successful crowdsourcing relies on efficient utilization and engagement of participants.} We contribute to this literature by  designing online policies for effectively notifying volunteers while avoiding  overutilization.

{\bf Online Matching and Prophet Inequalities}:
Abstracting away from the motivating application, our work is related to the stream of papers on online stochastic matching and prophet inequalities. Given the scope of this literature, we highlight only recent advances and kindly refer the interested reader to \citet{mehta2013online} for an informative survey. A standard approach is to design online policies based on an offline solution (see, e.g., \citealt{feldman2009online, OnlineWeighted, manshadi2012online, jaillet2014online, wang2018online}, and \citealt{stein2019advance})
and to compare the performance of these policies to a benchmark such as the clairvoyant solution described in \citet{golrezaei2014real}. Our work builds on this approach by applying techniques from prophet matching inequalities and the magician's problem \citep{alaei2012online, alaei2014bayesian}.
Further, while the classic setting for online matching focuses on bipartite graphs in which one side is static and the other side arrives online, a stream of recent papers (motivated by various applications) study dynamic matching problems in non-bipartite graphs \citep{Mypaper2,   ItaiAmin, Mypaper1} or in bipartite graphs where both sides arrive/leave  over time \citep{Vijay, aouad2019dynamic, TwoSidedProphet, Chiwei}.
Our setup also deviates from the classic online bipartite matching setting: although volunteers do not arrive online, they can be in two states (active or inactive) which can be viewed as arrival/departure. In contrast with the aforementioned papers, which all have an exogenous arrival/departure dynamic, volunteers' states in our setting are endogenously determined.

Our work  also contributes to a growing literature on online allocation of reusable resources. 
In a novel setting, \citet{adam} studies pricing of reusable resources and shows that static pricing achieves surprisingly good performance. {\revcolor Our work complements their approach by considering matching in a setting without prices.}
{\revcolor Closest to our framework are the innovative papers of \citet{dickerson2018allocation}, \citet{Rad2019}, \citet{gong2019online}, and \citet{rusmevichientong2020dynamic}. The former designs an adaptive policy to address an online stochastic matching problem in a setting with unit-capacity resources, while the latter three papers focus on online policies for resource allocation and assortment planning. We highlight three key 
ways in which our setting differs from these four papers.} (i) In our work, the platform's objective function is non-linear. Despite that, we only consider offline solutions that can be computed in polynomial time. {\revcolor In contrast, the four papers listed above either consider linear objectives or rely on an oracle {to solve an assortment optimization problem}.}
(ii) Volunteers---which represent the resources in our setting---can become unavailable without being matched (i.e., just through notification). (iii) We develop parameterized lower and upper bounds based on the minimum discrete hazard rate of the usage duration. {This} approach enables us to gain insight into the impact of characteristics of the usage duration distribution on the achievable bounds.

These crucial differences present new technical challenges which require us to develop new ideas 
in the design and analysis of our {\sdnpcolor policy} as well as in  setting a benchmark and computing an ex ante solution.
Our SN policy relies on solving individual-level DPs in order to sparsify an ex ante solution. The techniques used in the design and analysis of our SN policy build on ideas in \citet{alaei2012online}, \citet{alaei2014bayesian}, and \citet{Rad2019}. 
Further, 
our results rely on a primal-dual analysis, which is a powerful technique that has been used in other operational  problems (see, e.g., \citealt{ Jiawei}, \citealt{Gonzalo}, and \citealt{Levi2}). {\revcolor Additionally, in Appendix \ref{app:sdnp}, we design a second policy which scales down an ex ante solution by building on the approach of \citet{ma2018improvements} and \citet{dickerson2018allocation}.}




\section{Model}
\label{sec:model}

In this section, we formally introduce the 
\textit{online volunteer notification problem} that a volunteer-based crowdsourcing platform faces when deciding whom to notify for a task. 
{As part of} the problem definition, we highlight the platform's objective as well as the trade-off it faces due to the volunteers' adverse reactions to excessive notifications and the uncertainty in future tasks.  
Further, we define the measure of competitive ratio and establish a benchmark against which we compare the performance of any online policy.

The online volunteer notification problem  consists of a set of volunteers, denoted by $\set{V}$, and a set of task types, denoted by $\Don$.\footnote{For FRUS, a task represents a scheduled rescue (food donation) which has not been claimed in advance.} Volunteers (resp. task types) are indexed from $1$ to $|\V| = V$ (resp. $|\Don| = \D$). Over $T$ time steps, the platform solicits volunteers to complete a sequence of tasks. In particular, in each time step $t$, a task of type $\don$ arrives with known probability $\lambda_{\don,t}$. We assume at most one task arrives in each time step. Said differently, we assume $\sum_{\don = 1}^{\D} \lambda_{\don,t} \leq 1$ and with probability $1- \sum_{\don = 1}^{\D} \lambda_{\don,t} := \lambda_{0,t}$, no task arrives. Arrivals are assumed to be independent across time periods, but not \emph{within} each time period since at most one task arrives per period.

Whenever a task arrives, the platform 
{\revcolor must make an immediate and irrevocable decision about which volunteers to notify (if any), due to the time-sensitive nature of the tasks.}\footnote{{\revcolor In settings where decisions do not need to be made immediately, the benchmark which we establish continues to hold as does the lower bound achieved by our policy. However, additional strategies can be considered to take advantage of batching. See \citet{ItaiAmin} and \cite{feng2020batching} for two such examples.}}
Excessively notifying a volunteer may lead her to suffer from notification fatigue. To model this behavior in a general form, we assume that a volunteer can be in two possible states:  {\em active} or  {\em inactive}. In the former state, the volunteer pays attention to the platform's notifications, whereas in the latter state, she is inattentive {\revcolor and unaffected by additional notifications}. Initially, each volunteer is active.\footnote{{\revcolor In Section \ref{sec:conclude}, we discuss how our results extend when volunteers are not initially active.}} However, after {\revcolor an active volunteer is} notified
she transitions to the inactive state {\revcolor (regardless of whether or not she responds positively to the notification, as described below)}, and she will only become active again in $Z > 0$  periods, where  $Z$ is independently drawn from a known inter-activity time distribution denoted by $g(\cdot)$. Mathematically, $\P{Z = \tau} = g(\tau)$.

Before proceeding, we point out that similar modeling assumptions have been made in previous work. In particular, \citet{ata2016dynamic} models volunteer staffing for gleaning and assumes once a volunteer is utilized, she will go into a random repose period governed by an exponential distribution. Similarly, \citet{dickerson2019blood} focuses on blood donation and puts a constraint on the frequency with which a volunteer  can be notified, which is equivalent to assuming a deterministic inter-activity time. The latter strategy is also practiced in {many} FRUS locations.

To capture the minimum rate at which volunteers transition from inactive to active, we define the minimum discrete hazard rate of the inter-activity time distribution as follows:

\begin{definition}[Minimum Discrete Hazard Rate]
\label{def:hazard}
For a probability distribution $g(\cdot)$, the \emph{minimum discrete hazard rate} (MDHR) is given by $q = \min_{\tau \in \mathbb{N}} \frac{g(\tau)}{1-G(\tau-1)}$, {where $G(\cdot)$ denotes the corresponding CDF.}\footnote{By convention, if the fraction is $\frac{0}{0}$, we define it to be equal to 1.} {\revcolor We note that $q$ must be in the interval $[0,1]$.}
\end{definition}

{\revcolor If the inter-activity time distribution has an MDHR of $q$, then each volunteer will be active in each period with probability \emph{at least} $q$, regardless of the notification decisions made in the past. {Thus, we would intuitively expect that the cost of making a ``bad'' online decision diminishes as $q$ increases.
As we will show later,  both our lower bound (achieved by our policy) and our upper bound are increasing in $q$, which is aligned with this intuition (see Figure \ref{fig:hardness}).}}
We further highlight that a large value of $q$ is a sufficient condition to ensure that volunteers' activity level is high. For example, if $g(\cdot)$ is a geometric distribution, $q$ is the same as its success probability. {However, a small value of $q$ does not imply inactive volunteers:} if $g(2) = 1$, i.e., if the inter-activity times are deterministic and equal to 2 periods, then $q=0$ but volunteers are quite active.

If an  active volunteer $v$ is notified about a task of type $\don$, she will respond with match probability $p_{v,\don}$, independently from all other volunteers. Thus the arriving task is completed if at least one notified volunteer responds.\footnote{{For the remainder of the paper, when we say a volunteer ``responds'' we mean that the volunteer responds \emph{positively}, i.e., she is willing to complete the task.}} 
If a task of type $\don$ arrives at time $t$ and if the subset of notified and active volunteers is given by $\set{U}$, then the task will be completed with probability $1 - \prod_{v \in \set{U}} (1-p_{v,\don})$. We highlight that this probability is monotone and submodular with respect to the set $\set{U}$.
In Section \ref{sec:data}, we describe how $p_{v,\don}$ can be estimated accurately in the FRUS setting by using historical data.

As mentioned earlier, all volunteers are initially active. The platform knows the arrival rates $\lambda_{\don, t}$, the match probabilities $p_{v, \don}$, and the inter-activity time distribution $g(\cdot)$, but it does not observe volunteers' states. For any instance {$\set{I}$} of the online volunteer notification problem {where} $\set{I} = \big(\{\lambda_{\don,t}:\don \in [\D], t \in [T]\}, \{p_{v,\don}:v \in [V], \don \in [\D]\}, g \big)$,\footnote{For ease of notation, for any $a \in \mathbb{N}$, we use $[a]$ to refer to the set $\{1, 2, \dots, a\}$.} the platform's goal is to employ an online policy that maximizes the expected number of completed tasks. {\revcolor {This problem is a generalization of online stochastic bipartite matching, and it is intractable to solve optimally.
Indeed, even in the special case with no reusability and no subset selection, it is PSPACE hard to design a $(1-\epsilon)$-approximation of the optimal online policy \citep{papadimitriou2021online}.}}

Consequently, in order to evaluate an online policy, we
compare its performance to that of a clairvoyant solution that knows the entire sequence of arrivals in advance as well as volunteers' states in each period. However, {the clairvoyant solution} does not know \textit{before} notifying a volunteer how long her period of inactivity will be.
Two observations enable us to upper bound the clairvoyant solution with a polynomially-solvable program. First, note that if the clairvoyant solution notifies a subset of volunteers $\set{U}$ about a task of type $\don$, the probability of completing that task is 
\begin{align*}
    1 - \prod_{v \in \set{U}}(1-p_{v,\don}) \leq \min\Big\{ \sum_{v \in \set{U}}p_{v,\don} , 1\Big\}
\end{align*}

In words, we can upper bound the success probability of a subset $\set{U}$ with a piecewise-linear function that is the minimum of the expected total number of volunteer responses and $1$. {Second, recall that inactive volunteers ignore notifications. Thus, we can assume that the clairvoyant solution only notifies active volunteers, who will then become inactive for a random amount of time according to the inter-activity time distribution.} 
As a consequence, we can upper bound the clairvoyant solution via 
the program below,
which we denote by (LP).

\medskip
\begin{mdframed}
\begin{align}
\mathbf{LP}_\set{I} = \text{max}_{\mathbf{x}} \quad \quad \quad & \sum_{t=1}^T \sum_{\don=1}^{\D}   \lambda_{\don,t} \min\Big\{ \sum_{v=1}^V x_{v,\don,t}p_{v,\don}, 1\Big\}& \tag*{(LP)$^*$} \label{LP} \\
\text{s.t.} \ \ \ \qquad \quad 
    &0\leq  x_{v,\don,t}\leq 1  &\forall v, \don, t \label{eq:lpcon1} \\
    &1 \geq  \sum_{\tau =1}^t \sum_{\don=1}^{\D} \lambda_{\don,\tau} x_{v,\don,\tau} (1-G(t-\tau))  &\forall v, t \label{eq:lpcon2}
\end{align}
\rule{15.6cm}{.7pt}
\smallskip
\footnotesize{
\noindent $^*$ With a slight abuse of terminology, we refer to this program with a piecewise-linear objective as (LP). To express the above program as a linear program, we would replace $\min\{ \sum_{v=1}^V x_{v,\don,t}p_{v,\don}, 1 \}$ with $\sum_{v=1}^V x_{v,\don,t}p_{v,\don}$ in the objective. Then, we would add one constraint for each $\don, t$ pair ensuring that $\sum_{v=1}^V x_{v,\don,t}p_{v,\don} \leq 1$.} 
\end{mdframed}
\medskip

The decision variables $x_{v,\don,t}$ represent the probability of notifying volunteer $v$ when a task of type $\don$ arrives at time $t$. 
Constraint \eqref{eq:lpcon1} ensures that $x_{v,\don,t}$ is a valid probability. 
Constraint \eqref{eq:lpcon2} places limits on the frequency with which volunteers can be notified according to the inter-activity {\revcolor time} distribution. In particular, note that constraint \eqref{eq:lpcon2} can be written as
\begin{align*}
\sum_{\tau =1}^t \sum_{\don=1}^{\D} \lambda_{\don,\tau} x_{v,\don,\tau} \leq  1 + \sum_{\tau =1}^t \sum_{\don=1}^{\D} \lambda_{\don,\tau} x_{v,\don,\tau} G(t-\tau)
\end{align*}
The left hand side represents the expected total number of times a volunteer has been notified. {\revcolor The right hand side represents the {volunteer's initial active state plus the expected number of times the volunteer transitions from inactive to active.}}
Recall that any volunteer $v$ is initially active. 
Thus, the clairvoyant solution can notify her once. However, it will not notify  volunteer $v$ for the second time until she returns to the active state. Repeating this for all subsequent notifications up to time $t$ shows that
{\revcolor the notifications sent by the clairvoyant solution must respect the inter-activity time distribution in expectation, i.e.,}
the clairvoyant solution must meet constraint \eqref{eq:lpcon2}. 
We highlight that no online policy can achieve the optimal objective of (LP), even for instances with deterministic inter-activity times.
For ease of reference, in the following, we define the set of all feasible solutions to (LP). {Such a definition proves helpful in the rest of the paper.}
\begin{definition}[Feasible Set]
\label{def:P}
For any  $\mathbf{x} \in \mathbb{R}^{V \times \D \times T}$,
$\mathbf{x} \in \set{P}$ if {and only if} it satisfies
constraints \eqref{eq:lpcon1} and \eqref{eq:lpcon2}. 
\end{definition}

 The following proposition, which we prove in Appendix \ref{proof:LP}, establishes the relationship between the clairvoyant solution and $\mathbf{LP}_\set{I}$:

\begin{proposition}[Upper Bound on the Clairvoyant Solution] \label{prop:LP}
{For any instance $\set{I}$ of the online volunteer notification problem, $\mathbf{LP}_\set{I}$  is an upper bound on its clairvoyant solution}.
\end{proposition}

In light of Proposition \ref{prop:LP}, we use $\mathbf{LP}_\set{I}$ as a benchmark against which we compare the performance of any policy. Consequently, we define the competitive  ratio of an online policy as follows:\footnote{{\revcolor In the same spirit as \citet{golrezaei2014real, Rad2019}, and \citet{ma2020dynamic}, we define the competitive ratio relative to a Bayesian expected linear program benchmark as opposed to the exact clairvoyant solution.}}

 \begin{definition}[Competitive Ratio]
 \label{def:compratio}
 An online policy is $c$-competitive for the online volunteer notification problem if for any instance $\set{I}$, we have: $\mathbf{POL}_{\set{I}} \geq c \mathbf{LP}_{\set{I}}$, where $\mathbf{POL}_{\set{I}}$ represents the expected number of completed tasks by the online policy for instance ${\set{I}}$.
 \end{definition}

We will use the competitive ratio as a way to quantify the performance of an online policy. {\sdnpcolor For our SN policy} (presented in the following section), the competitive ratio is parameterized by the MDHR, $q$, and it improves as $q$ increases.

\section{Policy Design and Analysis}
\label{sec:algos}

In this section, we present and analyze {\sdnpcolor our SN policy for the online volunteer notification problem. This policy is randomized and relies on a fractional solution that} we compute ex ante using the instance primitives. Thus, we begin this section by introducing the ex ante solution in Section \ref{subsec:offline}. We then proceed to describe {\sdnpcolor our policy and analyze its competitive ratio in Section \ref{subsec:alg2}.}

\subsection{Ex Ante Solution}
\label{subsec:offline}
As stated in Section \ref{sec:intro}, {\sdnpcolor our online policy relies} on an ex ante solution which we  denote  by $\mathbf{x^*} \in [0,1]^{V \times \D \times T}$.  Given our benchmark, we focus our attention on solutions that are feasible in (LP), i.e., $\mathbf{x^*} \in \set{P}$ (see Definition \ref{def:P}). Clearly, $\mathbf{x^*_{LP}}$---the solution to (LP) in Section \ref{sec:model}---is a potential ex ante solution. However, in practice, such a solution {can} prove ineffective because it does not take into account the diminishing returns of notifying an additional volunteer about a task. As a result, it may ignore some tasks while notifying an excessive number of volunteers about others {(e.g., see the discussion and examples in Appendix \ref{ex:xstarcandidates})}.

{\revcolor Suppose that volunteers will always be active as long as the notifications sent to them respect the inter-activity time distribution in expectation, as given by constraint \eqref{eq:lpcon2}. In other words, as long as} 
$\mathbf{x} \in \set{P}$, suppose volunteers are always active when notified. Then if we notify each volunteer independently according to $\mathbf{x}$, the expected number of completed tasks would be:\footnote{Since a task can only be completed if one arrives, we limit all sums to task types indexed from $1$ to $\D$.}
\begin{align}
\label{eq:fdef1}
f(\mathbf{x}):= \sum_{t=1}^T \sum_{\don=1}^{\D} \lambda_{\don,t} \Big(1 - \prod_{v=1}^{V} (1- x_{v,\don,t}p_{v,\don}) \Big).
\end{align}
Because $\mathbf{x^*_{LP}}$ is the optimal solution of a piecewise-linear objective, it ignores the submodularity in  $f(\mathbf{x})$.\footnote{We remark that we design our online {\sdnpcolor policy such that it achieves} a constant factor of $f(x)$ {as defined in \eqref{eq:fdef1}.}}
In light of this intuition, we introduce two other candidates that can be computed in polynomial time. First, we aim to find the feasible point that maximizes $f(\cdot)$. We denote this optimization problem by \eqref{MRG} which stands for {\em Always Active}. Even though \eqref{MRG} is $NP$-hard \citep{bian2017guaranteed},  
simple polynomial-time algorithms such as the variant of the Frank-Wolfe algorithm described below (proposed in \citealt{bian2017guaranteed}) work well in practice. The algorithm iteratively maximizes a linearization of $f(\mathbf{x})$ and returns an average of feasible solutions, which therefore must be feasible. We denote the output of this algorithm by $\mathbf{x^*_{AA}}$ and use it as another candidate for the ex ante solution.

\medskip

\noindent
\begin{minipage}
{.23\linewidth}
\begin{mdframed}
\vspace{15.0pt}
\begin{equation}
  \max_{\mathbf{x} \in \mathcal{P}} f(\mathbf{x}) \label{MRG} \tag{AA}
\end{equation}
\vspace{15.0pt}
\end{mdframed}
\end{minipage}%
\begin{minipage}{.75\linewidth}
\begin{mdframed}
\underline{Approximating \ref{MRG} via Frank-Wolfe variant with step size $1/\FWstep$:}
\smallskip
\begin{enumerate}
\item Set $\mathbf{x}^0 = \mathbf{0}$.
    \item \textbf{For} $i$ from $1$ to $\FWstep$:
    \begin{enumerate}
        \item Solve $\mathbf{y}^i =\text{argmax}_{\mathbf{x} \in \set{P}} \langle \mathbf{x}, \nabla f(\mathbf{x}^{i-1}) \rangle$
        \item Set $\mathbf{x}^i = \mathbf{x}^{i-1} +\frac{1}{\FWstep} \mathbf{y}^i$
    \end{enumerate}
\item Return $\mathbf{x}^\FWstep$
\end{enumerate}
\end{mdframed}
\end{minipage}

\medskip

Note that the expected number of completed tasks, as defined in \eqref{eq:fdef1}, jointly depends on the contributions of all volunteers. 
This property makes optimizing such an objective challenging. Further, when assessing any online policy in this setting, jointly analyzing volunteers’ contributions  while keeping track of the joint distribution of their states (active or inactive) is prohibitively difficult.
We overcome this challenge by defining the following artificial priority scheme among volunteers which enables us to ``decouple'' the contributions of volunteers and find our last candidate for the ex ante solution.

\begin{definition}[Index-Based Priority Scheme]
\label{def:priority}
Under the index-based priority scheme, if multiple volunteers respond to a notification, the one with the smallest index completes the task.\footnote{Note that this priority scheme is without loss of generality, since in the online volunteer notification problem, \emph{all} active volunteers who receive a notification become inactive.}
\end{definition}

Following the index-based priority scheme allows us to define individual contributions for
each volunteer as shown in the following lemma (proven in Appendix \ref{proof:falg}).
\begin{lemma}[Volunteer Priority-Based Contributions]
\label{lem:falg}
For any $\mathbf{x} \in [0,1]^{V \times \D \times T}$, $f(\mathbf{x}) = \sum_{v = 1}^{V} f_v(\mathbf{x})$
 where $f(\cdot)$ is defined in \eqref{eq:fdef1} and
 \begin{align}
 \label{eq:decouple}
 f_v(\mathbf{x}) := \sum_{t=1}^T \sum_{\don=1}^{\D} \lambda_{\don,t} \Big(\prod_{u < v}(1-p_{u, \don} x_{u, \don, t})\Big)  p_{v, \don} x_{v, \don, t}.
\end{align}
\end{lemma}

Once again, suppose for a moment that volunteers are always active. Then for any $v \in [V]$, the term $\left(\prod_{u < v}(1-p_{u, \don} x_{u, \don, t})\right)  p_{v, \don} x_{v, \don, t}$ in \eqref{eq:decouple} 
represents the probability that under the index-based priority scheme, volunteer $v$ is the lowest-indexed volunteer to respond positively to a notification about a task of type $\don$ at time $t$.
Further, this term only depends on the fractional solution of volunteers with lower index than $v$. Thus, if we treat $x_{u, \don, t}$ as fixed for $1 \leq u <v$, then $\left(\prod_{u < v}(1-p_{u, \don} x_{u, \don, t})\right)  p_{v, \don} x_{v, \don, t}$ is linear in $x_{v, \don, t}$. In light of these observations, we define our last candidate as the solution of a  series of linear programs in which volunteers sequentially maximize their individual contributions in the order of their priority. This is summarized in the program \eqref{SQ}.

The separate but sequential nature of these programs leads to an efficiently-computable solution which takes into account the diminishing returns from notifying multiple volunteers.  To be specific, for a given volunteer $v$, the program \eqref{SQ} uses the solutions from previous iterations denoted by $x^{SQ}_{u,\don,t}$ for $u \in [v-1]$, $\don \in [\D]$, and $t \in [T]$.
As a result, the objective of \eqref{SQ} is a linear function of its decision variables, i.e., the $x_{v, \don, t}$ variables. Thus, \eqref{SQ} is a linear program, and its objective incorporates the externalities  imposed by  lower-indexed volunteers. We denote the solution to these $V$ sequential linear programs as $\mathbf{x^*_{SQ}}$.\footnote{The vector $\mathbf{x^*_{SQ}}$ consists of variables $x^{SQ}_{u,\don,t}$ for $u \in [V]$, $\don \in [\D]$, and $t \in [T]$.} Finally, we remark that the above decoupling idea proves helpful in both the design and analysis of our online {\sdnpcolor policy}. 

\medskip
\begin{mdframed}
\textbf{For} $v$ from $1$ to $V$:
\begin{align}
\max_{\{x_{v,\don,t}:\don \in [\D], t \in [T]\}} \ \  \quad \quad & \sum_{t=1}^T \sum_{\don=1}^{\D}  \lambda_{\don,t} \left(\prod_{u < v}(1-p_{u, \don} x^{SQ}_{u, \don, t})\right)  p_{v, \don} x_{v, \don, t}& \tag{SQ-$v$} \label{SQ} \\
\text{subject to}  \qquad \ \
   0 \leq &x_{v,\don,t} \leq 1 &\forall \don, t \nonumber \\ 
   1 \geq & \sum_{\tau =1}^t \sum_{\don=1}^{\D} \lambda_{\don,\tau} x_{v,\don,\tau} (1-G(t-\tau))  &\forall t \nonumber
\end{align}
\end{mdframed}
\medskip
Having three candidates, we define \begin{equation}
    \mathbf{x^*}:= \text{argmax}_{\mathbf{x} \in \{\mathbf{x^*_{LP}}, \mathbf{x^*_{AA}}, \mathbf{x^*_{SQ}} \}} f(\mathbf{x}) \label{eq:xstar}
\end{equation} 
The following proposition, which we prove in Appendix \ref{proof:fxstar}, establishes a lower bound on $f(\mathbf{x^*})$ based on the benchmark $\mathbf{LP}$. 

\begin{proposition}[{Lower Bound on Ex Ante Solution}] \label{prop:fxstar}
For $\mathbf{x^*}$ defined in \eqref{eq:xstar},  $$f(\mathbf{x^*})  \geq (1-\frac{1}{e})\mathbf{LP}.$$
\end{proposition}

The above worst case ratio is achieved by the ratio of $f(\mathbf{x^*_{LP}})$ to $\mathbf{LP}$, and it is tight. However, we stress that $\mathbf{x^*_{AA}}$ and  $\mathbf{x^*_{SQ}}$ can provide significant improvements. A simple example illustrating this point can be found in Appendix \ref{ex:xstarcandidates}, along with an example demonstrating that $f(\mathbf{x^*_{LP}})$ can be strictly greater than $f(\mathbf{x^*_{AA}})$ and $f(\mathbf{x^*_{SQ}})$. 
These examples show that none of the three solutions is universally dominant (or dominated); as such, we define $\mathbf{x^*}$ to be the maximum of the three.

We conclude this section by noting that an online policy which directly follows $\mathbf{x^*}$ (i.e., a policy that at time $t$, upon arrival of $\don$, notifies volunteer $v$ independently with probability $x^*_{v,\don,t}$) achieves a competitive ratio of at most $q$, as shown in Proposition \ref{prop:ubexante} in Section \ref{subsec:upperboundexante}. This hardness result stems from the fact that $\mathbf{x^*}$ ``respects'' the inactivity period of volunteers only in expectation. Consequently, under a policy of directly following $\mathbf{x^*}$, it is possible that volunteers are inactive when high-value tasks (e.g. tasks where the match probability is close to $1$) arrive because they were notified earlier (according to $\mathbf{x^*}$) for low-value tasks. {\sdnpcolor Therefore, we develop a policy based on a sparsification of the ex ante solution, which we describe and analyze in the subsequent section.}

\subsection{Sparse Notification Policy}
\label{subsec:alg2}

{\revcolor Before we present the sparse notification (SN) policy---which earns its name by sparsifying the ex ante solution---momentarily consider a simpler policy which proportionally scales down the ex ante solution. Though intuitive, such a policy 
relies exclusively on the ex ante solution to resolve the trade-off between the immediate reward of notifying a volunteer and saving her for a future arrival.
Rather than considering each decision individually, it adjusts the ex ante solution on an aggregate level, which can be suboptimal: even in the last period $T$, such a policy {follows a scaled-down version of $\mathbf{x^*}$}
  despite getting no benefit from saving a volunteer for a future arrival.}

{To more accurately resolve this trade-off, in designing the SN policy, we utilize} the ex ante solution and the index-based priority scheme (see Definition \ref{def:priority}) {to} formulate a sequence of one-dimensional DPs whose optimal value will serve as a lower bound on the contribution of each volunteer according to her priority (as shown in Lemma \ref{lem:contributionsalg2}). The solution of these DPs is a sparsified version of the ex ante solution $\mathbf{x^*}$. Namely, let us denote $\mathbf{\y}$ as the solution of the sequence of DPs. For any $v$, $\don$, and $t$, $\tilde{x}_{v,\don, t}$ is either $0$ or $x^*_{v,\don, t}$.
Equipped with $\mathbf{\y}$, {which we compute in advance,} the SN policy probabilistically follows $\mathbf{\y}$ in the online phase. Our DP formulation and its analysis builds on the framework developed in  \citet{alaei2012online} and \citet{alaei2014bayesian}, which is also used in \citet{Rad2019}. 

Next we describe the DP formulation. Consider volunteer $v \in [V]$ and suppose we have already solved the first $(v-1)$ DPs. Thus we have $\{\tilde{x}_{u,\don,t}: u \in [v-1], \don \in [\D], t \in [T]\}$. Let us denote the value-to-go of the DP at time $t$ by $J_{v,t}$. As mentioned above, we define the DP such that 
 $J_{v,t}$ is  a lower bound on the expected number of tasks that volunteer $v$ will complete between $t$ and $T$ when active at time $t$ under the SN policy and the index-based priority scheme (we prove this assertion in Lemma \ref{lem:contributionsalg2}). Clearly $J_{v,T+1} =0$. To specify $J_{v,t}$ for $t \in [T]$, we first define $v$'s reward at time $t$ when active and notified about a task of type $\don$ as follows: 
\begin{equation}
r_{v, \don, t} := p_{v,\don} \prod_{u =1}^{v -1}(1-\y_{u,\don,t}p_{u,\don})\footnote{We emphasize that this is not the actual reward, i.e., it is not the probability that volunteer $v$ completes a task of type $\don$ under the index-based priority scheme. However, it is a lower bound, as shown in the proof of Lemma \ref{lem:contributionsalg2}.} \label{eq:rewards}
\end{equation}
{The only two actions available when a task of type $\don$ arrives at time $t$ are} to notify $v$ with probability $x^*_{v,\don,t}$ or to not notify $v$. Thus when deciding on the optimal action, we compare the (current and future) reward of notifying $v$ now to the reward of saving her for the next period. Formally,
\begin{align}
\label{eq:DP:sol}
    \tilde{x}_{v, \don, t} = x^*_{v, \don, t} \I{r_{v,\don,t} + \sum_{\tau = t+1}^T g(\tau - t)\J_{v, \tau} \geq \J_{v,t+1}}
\end{align}
The term in the indicator on the left hand side is the reward of notifying $v$ in the current period $t$, which consists of two parts: (i) the immediate reward we get from notifying $v$---which will make her inactive for $Z$ periods---and (ii) the future reward once she becomes active again. The right hand side within the indicator simply represents the reward when $v$ is not notified {and remains active in period $t+1$}. Given \eqref{eq:rewards}, \eqref{eq:DP:sol}, {and $J_{v, T+1} = 0$}, we can iteratively compute $\{J_{v,t}; t \in [T] \}$ as follows:\footnote{{To compute the value-to-go at time $t$, we must sum over \emph{all} possible arrivals, including tasks of type $0$ (i.e., no task). By convention, we set variables associated with tasks of type $0$ (e.g. $x_{v,0,t})$ to be $0$.}}
\begin{align}
    J_{v,t} = \sum_{\don=0}^{\D} \lambda_{\don, t}\Big((1-\y_{v, \don, t})\J_{v,t+1} + \y_{v,\don,t}\Big( r_{v,\don,t} + \sum_{\tau = t+1}^{T} g(\tau - t)\J_{v, \tau}\Big) \Big) \label{eq:jv}
\end{align}
The formal definition of our policy is presented in Algorithm \ref{alg:two}. In the rest of this section, we analyze the competitive ratio of the SN policy. Our main result is the following theorem:

\begin{algorithm}
	\textbf{Offline Phase}:
	    \begin{enumerate}
	        \item Compute $\mathbf{x^*}$ according to \eqref{eq:xstar}
	    \item \textbf{For} all $v \in [V]$:
	    \begin{enumerate}
	        \item \textbf{For} all $\don \in [\D]$ and $t \in [T]$, {compute $r_{v,\don, t}$ according to \eqref{eq:rewards}}
	        \item Set $\J_{v,T+1} = 0$
	        \item \textbf{For} $t =T$ to $t = 1$ :
	        \begin{enumerate}
	        	\item \textbf{For} all $\don \in [\D]$, {compute $\tilde{x}_{v, \don, t}$ according to \eqref{eq:DP:sol}}
	            \item {Compute $J_{v,t}$ according to \eqref{eq:jv}}
	        \end{enumerate}
	    \end{enumerate}
	\end{enumerate}
	
	\textbf{Online Phase}:
	\begin{enumerate}
	\item \textbf{For} $t$ from $1$ to $T$:
	\begin{enumerate}
	\item \textbf{If} a task of type $\don$ arrives in time $t$, then:
	\begin{enumerate}
	    \item {\revcolor \textbf{For} $v \in [V]$:}
	\begin{itemize}
	\item {Notify $v$ with probability $\y_{v,\don, t}$}
	\end{itemize}
	\end{enumerate}
	\end{enumerate}
	\end{enumerate}
	\caption{Sparse Notification (SN) Policy}
	\label{alg:two}
\end{algorithm}

\begin{theorem}[Competitive Ratio of the Sparse Notification Policy] \label{thm:alg2}
Suppose that the MDHR of the inter-activity time distribution is $q$. Then the sparse notification policy, defined in Algorithm \ref{alg:two}, is $\frac{1}{2-q}(1-\frac{1}{e})$-competitive.
\end{theorem}

{\revcolor Theorem \ref{thm:alg2} implies that the SN policy is $\frac{1}{2}(1-\frac{1}{e})$ competitive, regardless of the inter-activity time distribution, which can be shown by taking an infimum over all $q \in [0,1]$. The competitive ratio of the SN policy improves as $q$ increases, even though the policy does not directly make use of $q$ in its design.} The proof of Theorem \ref{thm:alg2} consists of two main lemmas. First, in the following lemma, we lower bound the contribution of each volunteer $v$ by $J_{v,1}$: 

\begin{lemma}[Volunteer Priority-Based Contribution under the SN Policy]
\label{lem:contributionsalg2}
Under the index-based priority scheme (in Definition \ref{def:priority}) and the SN policy, the contribution of volunteer $v \in [V]$, i.e., the expected number of tasks she completes, is at least $J_{v,1}$, where 
$J_{v,1}$ is defined in \eqref{eq:jv}.
\end{lemma}

\begin{proof}{Proof:}
The proof of Lemma \ref{lem:contributionsalg2} consists of two parts. \textbf{Part (i):} First, we prove that $r_{v, \don, t}$ is a lower bound on the probability that a volunteer $v \in [V]$ completes a task of type $\don \in [\D]$ when it arrives at time $t \in [T]$, conditional on being notified and active, under the SN policy and the index-based priority scheme. {\textbf{Part (ii):} Then we show that {under such a policy and priority scheme,} the expected number of tasks volunteer $v$ will complete between $t$ and $T$ when active at $t$ must be at least $J_{v,t}$.}
 
\textbf{Proof of Part (i):}
Without loss of generality, we focus on a particular arrival $\don$ at a particular time $t$. {Let us define 
$\rtwo_{v, \don , t}$ 
as the probability that volunteer $v$ completes a task of type $\don$ when it arrives at time $t$, conditional on $v$ being notified and active at time $t$, under the SN policy and the index-based priority scheme. We will show that $r_{v,\don , t}$ as defined in \eqref{eq:rewards} is a lower bound on 
$\rtwo_{v, \don , t}$.
} 
When notified and active, a volunteer $v \in [V]$ responds with probability $p_{v,\don}$. Any other lower-indexed volunteer $u \in [v-1]$ is notified with probability $\y_{u, \don, t}$ under the SN policy. If active, she will respond with probability $p_{u,\don}$. Since these are both independent from $v$'s response, the probability that $u$ responds conditional on $v$ responding must be less than $\y_{u, \don, t}p_{u, \don}$. {Repeating this argument jointly for all $u < v$, we see that $\rtwo_{v, \don , t}$, i.e., the probability} that $v$ completes the task when active and notified---which happens when she is the lowest indexed volunteer to respond---must be at least $p_{v,\don}\prod_{u =1}^{v-1}(1-\y_{u, \don, t}p_{u, \don})$. Noting that this is equivalent to the definition of $r_{v,\don, t}$ completes the first part of the proof, namely that $\rtwo_{v, \don , t} \geq r_{v, \don , t} $.

\textbf{Proof of Part (ii):}
{Let us define $\Jtwo_{v,t}$ as the expected number of tasks $v$ will complete between $t$ and $T$ when active at $t$ under the SN policy and the index-based priority scheme.
 We will show via total backward induction that $\Jtwo_{v,t} \geq J_{v,t}$. Clearly, this is true with equality for $\Jtwo_{v,T+1} = 0$. Now suppose that this inductive hypothesis holds for all $t \geq k+1$. We will show that for $t=k$, $\Jtwo_{v,k} \geq J_{v,k}$. By construction, we have
 \begin{align}
 \Jtwo_{v,k} = \sum_{\don=0}^{\D} \lambda_{\don, k}\left((1-\y_{v, \don, k})\Jtwo_{v,k+1} + \y_{v,\don,k}( \rtwo_{v,\don,k} + \sum_{\tau = k+1}^{T} g(\tau - k)\Jtwo_{v, \tau}) \right) \nonumber
 \end{align}

In words, the expected number of tasks $v$ will complete between $t$ and $T$ when active at $t$ under the SN policy and the index-based priority scheme can be computed in the following way: fixing an arrival of a task of type $\don$ at time $k$, if $v$ is not notified, she will complete $\Jtwo_{v,k+1}$ tasks (in expectation) in the future. If $v$ is notified, she will complete this task with probability $\rtwo_{v, \don , k}$. The expected number of tasks she will complete in the future is the expected number of tasks she will complete after becoming active again, as given by the term $\sum_{\tau = k+1}^{T} g(\tau - k)\Jtwo_{v, \tau}$. 
Summing over all task types $s \in [S] \cup \{0\}$, we get $\Jtwo_{v,k}$. Using our inductive hypothesis and part (i) of the proof, we have 
 \begin{align}
 \Jtwo_{v,k} \geq \sum_{\don=0}^{\D} \lambda_{\don, k}\left((1-\y_{v, \don, k})\J_{v,k+1} + \y_{v,\don,k}( r_{v,\don,k} + \sum_{\tau = k+1}^{T} g(\tau - k)\J_{v, \tau}) \right) \label{eq:contribalg2eq1}
 \end{align}
 Noting that the right hand side of \eqref{eq:contribalg2eq1} is exactly the definition of $J_{v,k}$ according to \eqref{eq:jv}, we have shown $\Jtwo_{v,k} \geq J_{v,k}$, which completes the proof by induction. 
 
This implies that the expected number of tasks volunteer $v$ will complete between $1$ and $T$ when active at period $1$ under the SN policy and the index-based priority scheme must be at least $J_{v,1}$. Because $v$ is by definition active in period $1$, we have completed the proof of Lemma \ref{lem:contributionsalg2}.}
\Halmos
\end{proof}

 


The second main step in the proof of Theorem \ref{thm:alg2} is to compare $J_{v,1}$ to the benchmark $\mathbf{LP}$. In order to do so, we follow the dual-fitting approach of  \citet{alaei2012online}. {\revcolor In particular, given the inter-activity time distribution, we set up a linear program to find the ``worst'' possible combination of per-stage rewards (denoted by the decision variables $\bar{\mathbf{r}}=\{{\revcolor \bar{r}}_{v,\don,t}:\don \in [\D], t \in[T]\}$) that give rise to the minimum possible value of the initial value-to-go of the DP (denoted by decision variable ${\revcolor \bar{J}}_{v,1}$). }
In the LP formulation, the first two sets of constraints follow from the DP definition. Note that the value of ${\revcolor \bar{J}}_{v,1}$ will crucially depend on the values of per-stage rewards through $\sum_{t=1}^T \sum_{\don=0}^{\D} \lambda_{\don,t} {\revcolor \bar{r}}_{v, \don, t} x^*_{v,\don,t}$, e.g., if ${\revcolor \bar{r}}_{v,\don,t} = 0$ for all $v \in [V]$, $\don \in [\D]$, and $t \in [T]$, then ${\revcolor \bar{J}}_{v,1} = 0$. This motivates the final constraint, which provides a constant against which we can compare ${\revcolor \bar{J}}_{v,1}$.
Finding the optimal solution to this LP proves to be difficult. Instead we find a feasible solution to its dual (the LP and its dual are presented in Table \ref{table:lpdual}). The following lemma uses this dual solution to establish a lower bound on the initial value-to-go of the DP, regardless of the per-stage rewards.

\begin{table}
  \caption{The linear and dual programs used to provide a lower bound on $J_{v,1}$ for all $v \in [V]$.}
{\small $$\begin{array}{|clc|clc|}
\hline
&&&&&\\
\multicolumn{3}{|c|}{\text{(J-LP) uses variables } \mathbf{{\revcolor \bar{J}}}=\{{\revcolor \bar{J}}_{v,t}:t \in[T]\} }&\multicolumn{3}{c|}{\text{(Dual) uses variables } \bm{\alpha}=\{\alpha_t \geq 0 : t \in [T]\}, }  \\
\multicolumn{3}{|c|}{\text{and } \mathbf{{\revcolor \bar{r}}}=\{{\revcolor \bar{r}}_{v,\don,t}:\don \in [\D], t \in[T]\}. }&\multicolumn{3}{c|}{\bm{\gamma}=\{\gamma_t \geq 0 : t \in [T]\}, \text{ and } \mu.}  \\
&&&&&\\
\hline
&&&&&\\
\ \ \min_{\mathbf{{\revcolor \bar{J}}}, \mathbf{{\revcolor \bar{r}}}} & {\revcolor \bar{J}}_{v,1} & \text{ (J-LP) }&\ \ \max_{\bm{\alpha}, \bm{\gamma}, \mu}&c\mu&\text{ (Dual) }\\
&&&&& \\
\text{s.t.} &&&\text{s.t.}&&\\
&&&&& \\
{\revcolor \bar{J}}_{v,t}  \geq& {\revcolor \bar{J}}_{v,t+1}&\ (\alpha_t) \ &\gamma_1\leq&1-\alpha_1&\ ({\revcolor \bar{J}}_{v,1}) \ 
\\
&&&&& \\
{\revcolor \bar{J}}_{v,t}\geq& {\revcolor \bar{J}}_{v, t+1} + \sum_{\don=0}^{\D} \lambda_{\don,t} x^*_{v,\don,t} &&\gamma_t\leq&\gamma_{t-1}+\alpha_{t-1}-\alpha_t-& \\
&[{\revcolor \bar{r}}_{v,\don, t}-{\revcolor \bar{J}}_{v,t+1}+&&&\gamma_{t-1}\sum_{\don=0}^{\D} \lambda_{\don,t-1}x^*_{v,\don,t-1}+& \\
&\sum_{\tau = t+1}^{T} g(\tau - t){\revcolor \bar{J}}_{v, \tau})]&\ (\gamma_t) \ &&\sum_{t' = 1}^{t-1}\gamma_{t'} \sum_{\don=0}^{\D} \lambda_{\don,t'} x^*_{v,\don,t'}g(t-t')&\ ({\revcolor \bar{J}}_{v,2:T}) \ \\
&&&&& \\
\ \ c =& \sum_{t=1}^T \sum_{\don=0}^{\D} \lambda_{\don,t} {\revcolor \bar{r}}_{v, \don, t} x^*_{v,\don,t}&\ (\mu) \ &
\gamma_t\geq& \mu&\ (\mathbf{{\revcolor \bar{r}}}) \ 
\\
&&&&&\\
\hline
\end{array}$$}
\label{table:lpdual}
\end{table}

\begin{lemma}[Lower Bounding the Dynamic Program]
\label{lemma:factorrevealingLP}
Under the index-based priority scheme (see Definition \ref{def:priority}), for any $\mathbf{x} \in \set{P}$ and volunteer $v \in [V]$, we have $J_{v,1} \geq  \frac{1}{2-q}f_v(\mathbf{x})$
where $f_v(\mathbf{x})$ is defined in \eqref{eq:decouple}. 
\end{lemma}
\begin{proof}{Proof:}
First we show that
 for a particular volunteer $v \in [V]$, the solution to (J-LP) is a lower bound on the initial value-to-go ${\revcolor \bar{J}}_{v,1}$ which occurs when the per stage rewards are given by ${\revcolor \bar{r}}_{v, \don, t}$. To see this, we show that the first two sets of constraints in (J-LP) come from the iterative definition of the value-to-go, as given in equation \eqref{eq:jv}: \begin{align}
    {\revcolor \bar{J}}_{v,t} &= \sum_{\don=0}^{\D} \lambda_{\don, t} \max_{x_{v,\don,t} \in \{0, x^*_{v, \don, t}\}} \{ (1-x_{v,\don, t}){\revcolor \bar{J}}_{v,t+1} + x_{v,\don, t}({\revcolor \bar{r}}_{v,\don,t} + \sum_{\tau = t+1}^{T} g(\tau - t){\revcolor \bar{J}}_{v, \tau}) \} \nonumber \\
    &\geq \max\{{\revcolor \bar{J}}_{v,t+1}, \sum_{\don=0}^{\D} \Big[\lambda_{\don, t}(1-x^*_{v,\don, t}){\revcolor \bar{J}}_{v,t+1}+ x^*_{v,\don, t}({\revcolor \bar{r}}_{v,\don,t} + \sum_{\tau = t+1}^{T} g(\tau - t){\revcolor \bar{J}}_{v, \tau})\Big]\} \label{eq:LP:const:s1} \\
    &=\max\{{\revcolor \bar{J}}_{v,t+1}, {\revcolor \bar{J}}_{v,t+1}+\sum_{\don=0}^{\D} \lambda_{\don, t} x^*_{v,\don, t}({\revcolor \bar{r}}_{v,\don,t} -{\revcolor \bar{J}}_{v,t+1} + \sum_{\tau = t+1}^{T} g(\tau - t){\revcolor \bar{J}}_{v, \tau})\} \label{eq:LP:const:s2}
\end{align}

In the above, the first equality follows from the definition of $\y_{v, \don, k}$. The inequality in \eqref{eq:LP:const:s1} follows from setting the values of $x_{v, \don, k}$ to their extremes, i.e., $x_{v, \don, k} = 0$ (which gives the first term inside the $\max$) and $x_{v, \don, t} = x^*_{v, \don, t}$ (which gives the second term inside the $\max$).  The last equality is a result of simplifying the  second term inside the $\max$. Note that \eqref{eq:LP:const:s2} implies that the first two constraints in (J-LP) must hold. The final constraint in (J-LP) {scales} the per-stage rewards while allowing for the ``worst'' possible combination. Together there are $2T + 1$ constraints, which will become the dual variables identified by the labels $\bm{\alpha}=\{\alpha_t \geq 0 : t \in [T]\}$, 
$\bm{\gamma}=\{\gamma_t \geq 0 : t \in [T]\}$, and $\mu$, respectively, 
in Table \ref{table:lpdual}. This leads to the dual program in (Dual).

Next, we show that the following is 
a feasible solution to this dual problem: $\mu = \frac{1}{2-q}$ and all constraints are tight, i.e. $\gamma_t = \mu$ for all $t \in [T]$, $\alpha_1 = 1-\mu$, and for $t \geq 2$, 
\begin{align}
    \alpha_t &= \alpha_{t-1} +\gamma_{t-1}-\gamma_t - \gamma_{t-1}\sum_{\don=0}^{\D} \lambda_{\don,t-1}x^*_{v,\don,t-1} + \sum_{t' =1}^{t-1}\gamma_{t'} \sum_{\don=0}^{\D} \lambda_{\don,t'} x^*_{v,\don,t'}g(t-t') \nonumber \\
    &=\alpha_{t-1} - \mu \left( \sum_{\don=0}^{\D} \lambda_{\don,t-1}x^*_{v,\don,t-1} -  \sum_{t' =1}^{t-1} \sum_{\don=0}^{\D} \lambda_{\don,t'} x^*_{v,\don,t'}g(t-t') \right) \label{eq:flp0} \\
    &= \alpha_1 - \mu \sum_{t' =1}^{t-1} \sum_{\don=0}^{\D} \lambda_{\don,t'} x^*_{v,\don,t'}(1-G(t-t')) \label{eq:flp1}\\
    &=\alpha_1 - \mu \sum_{t' =1}^{t-1} \sum_{\don=0}^{\D} \lambda_{\don,t'} x^*_{v,\don,t'}(1-G(t-1-t'))(1-\frac{g(t-t')}{1-G(t-1-t')})
    \label{eq:flp2} 
     \\
    &\geq  \alpha_1 - \mu \sum_{t' =1}^{t-1} \sum_{\don=0}^{\D} \lambda_{\don,t'} x^*_{v,\don,t'}(1-G(t-1-t'))(1-q) \label{eq:flp3}\\
    &\geq \alpha_1 - (1-q)\mu \label{eq:flp4} \\
    &= 1 - (2-q)\mu \nonumber
\end{align}
Equality in Line \eqref{eq:flp0} follows from plugging in $\gamma_{t} = \mu$ for all $t$.
Line \eqref{eq:flp1} comes from recursively plugging in the definition for $\alpha_{t-1}$ and rearranging terms. 
Line \eqref{eq:flp2} is a result of re-writing $(1-G(t-t'))$ as $(1-G(t-t'-1) - g(t-t'))$ and then factoring out $(1-G(t-t'-1))$.\footnote{{If $1-G(t - t'-1) = 0$, then we must also have $g(t-t') = 0$. Thus, in Line \eqref{eq:flp2}, we preserve the equality by following our convention that if the fraction is $\frac{0}{0}$, we define it to be equal to 1.}}
Line \eqref{eq:flp3} comes from applying the definition of the minimum hazard rate. Line \eqref{eq:flp4} uses the fact that $\mathbf{x^*} \in \set{P}$, which means it must satisfy constraint \eqref{eq:lpcon2} of (LP) at time $t-1$.

Finally, note that in this proposed solution, all the dual variables are non-negative: $0 \leq q \leq 1$, which ensures $\mu = \frac{1}{2-q} \in [1/2,1]$. Thus $\gamma_t = \mu \geq 0$, and $\alpha_1 \geq 0$.
As $\mu = \frac{1}{2-q}$, we must have $\alpha_t \geq 0$. 
{Since all constraints are tight and for all $t \in [T]$, $\alpha_t \geq 0$ and 
$\gamma_t \geq 0$}, this solution is feasible in (Dual). Therefore, by weak duality, we have 
\begin{align}
    {\revcolor \bar{J}}_{v,1} \geq \frac{1}{2-q} c = \frac{1}{2-q} \sum_{t=1}^T \sum_{\don=0}^{\D} \lambda_{\don,t} {\revcolor \bar{r}}_{v, \don, t} x^*_{v,\don,t}
    \label{eq:LP:last}
\end{align}
{\revcolor This lower bound holds for the worst-case combination of per-stage rewards $\mathbf{{\revcolor \bar{r}}}$, meaning that for any set of rewards $\mathbf{r}$, we must similarly have ${J}_{v,1} \geq \frac{1}{2-q} \sum_{t=1}^T \sum_{\don=0}^{\D} \lambda_{\don,t} {r}_{v, \don, t} x^*_{v,\don,t}$.} We finish the proof by showing that $r_{v,\don,t}$ (as defined in \eqref{eq:rewards}) is at least $\left(\prod_{u < v}(1-p_{u, \don} x^*_{u, \don, t})\right)  p_{v, \don}$. To see this, note that the term $\left(\prod_{u < v}(1-p_{u, \don} x_{u, \don, t})\right)  p_{v, \don}$ is decreasing in $x_{u,\don,t}$. Since $\y_{u, \don, t} \leq x^*_{u,\don,t}$, this implies $r_{v,\don,t} \geq \left(\prod_{u < v}(1-p_{u, \don} x^*_{u, \don, t})\right)  p_{v, \don}$. Plugging this back into \eqref{eq:LP:last}, we have:
$$J_{v,1} \geq \frac{1}{2-q} \sum_{t=1}^T \sum_{\don=0}^{\D} \lambda_{\don,t} \left(\prod_{u < v}(1-p_{u, \don} x^*_{u, \don, t})\right)  p_{v, \don} x^*_{v, \don, t} = {\revcolor \frac{1}{2-q}}f_v(\mathbf{x^*}),$$ which completes the proof of Lemma \ref{lemma:factorrevealingLP}.
\Halmos \end{proof}

Based on Lemma \ref{lemma:factorrevealingLP}, we know that each volunteer completes at least $\frac{1}{2-q}f_v(\mathbf{x^*})$ tasks in expectation. By linearity of expectations and Lemma \ref{lem:falg}, the expected total number of tasks completed by volunteers must be at least $\frac{1}{2-q}f(\mathbf{x^*})$. Since $f(\mathbf{x^*}) \geq (1-\frac{1}{e})\mathbf{LP}_\set{I}$ (see Proposition \ref{prop:fxstar}), it immediately follows that the SN policy is  $\frac{1}{2-q}(1-\frac{1}{e})$-competitive. This completes the proof of Theorem \ref{thm:alg2}. \Halmos

\section{Upper Bounds on Competitive Ratio}
\label{sec:hardness}

In this section, we provide upper bounds on the achievable performance of various policies for the online volunteer notification problem. We begin in Section \ref{subsec:upperboundany} by upper-bounding the competitive ratio of any online policy. Then, in Section \ref{subsec:upperboundexante}, we place an upper bound on the competitive ratio of the specific policy of directly following the ex ante solution as defined in \eqref{eq:xstar}.

\subsection{Upper Bound for Any Online Policy} 
\label{subsec:upperboundany}

 Like the lower bound achieved by our policies in Section \ref{sec:algos}, the upper bound we establish for any online policy is parameterized by the MDHR of the inter-activity time distribution, $q$. We highlight that our upper bound applies to \emph{all} online policies, even those that cannot be computed in polynomial time. The main result of this section is the following theorem: 

\begin{theorem}[Upper Bound on Achievable Competitive Ratio]
\label{thm:hardness}
Suppose the MDHR of the inter-activity time distribution is $q$ where $q \in [1/16,1] \cup \{1/n,  n \in \mathbb{N} \} \cup \{0 \} $.   Then no online policy can achieve a competitive ratio greater than $\kappa$, where for $q>0$
\begin{equation}
\label{eq:kappa}
    \kappa = \min \Big\{\frac{1}{2-q}, 1 + q - \frac{q(1-q)}{\log(\frac{1}{1-q})(1+q)}(1-e^{-1}) \Big \}
\end{equation}
and for $q = 0$, we have {\revcolor $\kappa = 0.334$.}\footnote{We remark that the condition imposed on  $q$ when $0 < q < 1/16$ is added for ease of presentation of the theorem statement as well as its proof. Relaxing the aforementioned condition amounts to modifying the second term in $\kappa$ by rounding any $q$ up to the closest $\hat{q} \in \{1/n,  n \in \mathbb{N} \}$ and slightly modifying  the instance in the proof. We omit these details for the sake of brevity. }
\end{theorem}

\begin{figure}[t]
 \centering
\includegraphics[scale = 0.6]{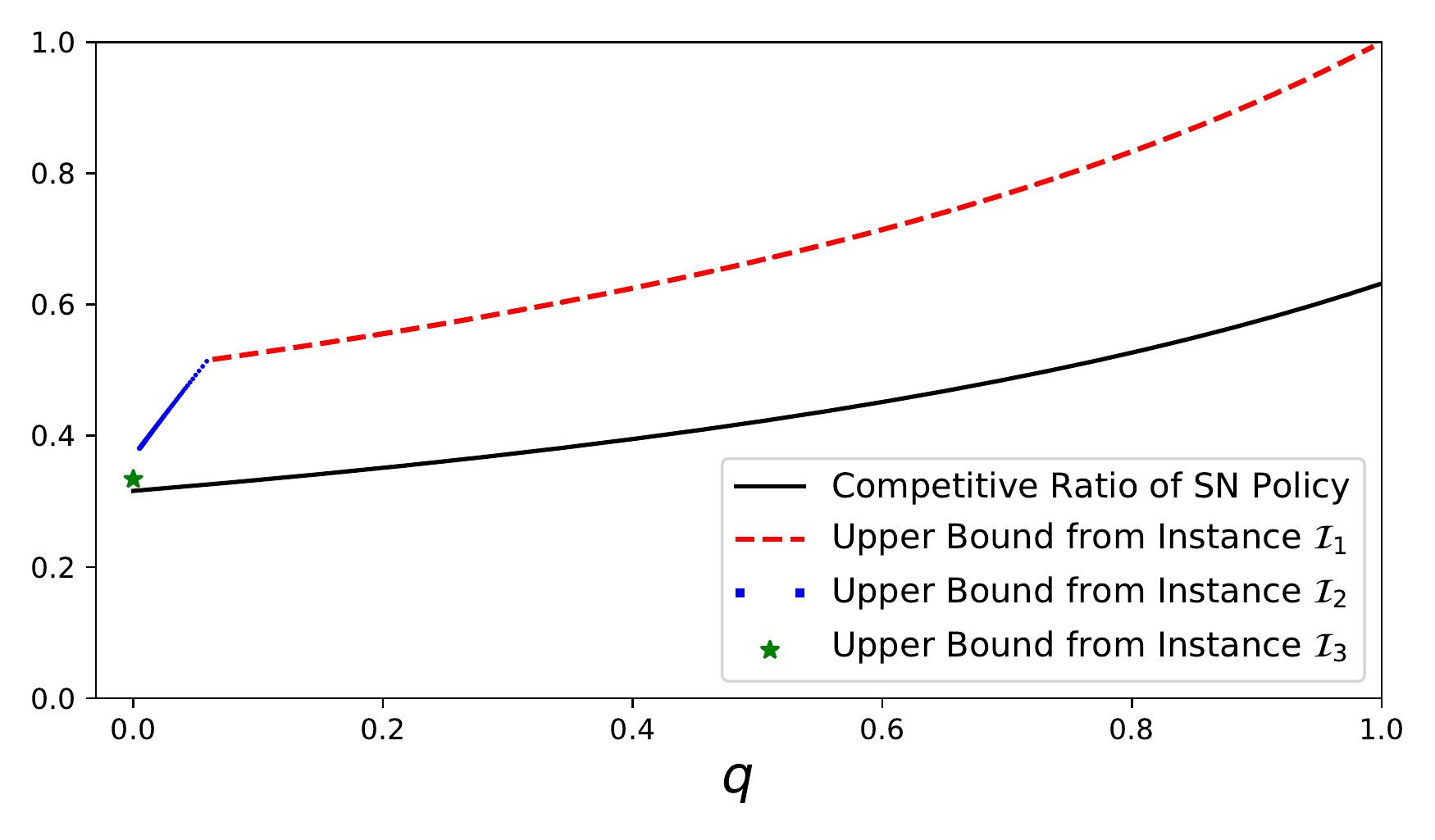}
\caption{Comparing the competitive ratio of our SN policy presented in Section \ref{sec:algos} (Theorem \ref{thm:alg2}) to an upper bound on the performance of any online policy (Theorem \ref{thm:hardness}) as a function of the MDHR, $q$.}
\label{fig:hardness}
\end{figure}

Figure \ref{fig:hardness} provides a summary of our lower and upper bounds on the achievable competitive ratio for the online volunteer notification problem as a function of $q$. We make the following observations based on the {theorem and accompanying} plot: 
(i) both the upper and lower bounds improve as $q$ increases, and (ii) the competitive ratio of our {\sdnpcolor SN policy is} fairly close to the upper bound when $q$ is small. However, the gap grows for larger values of $q$. 
The proof of Theorem \ref{thm:hardness} relies on analyzing the {\revcolor three instances described below. Instance $\set{I}_1$ attains the minimum when $q \in [1/16,1)$,  
instance $\set{I}_2$ attains it when $q \in \{1/n, n > 16, n \in \mathbb{N} \}$, and finally, instance $\set{I}_3$ attains it when $q=0$.}\footnote{For $q=1$, by definition no online policy can achieve a competitive ratio greater than $1$.}

\medskip 
{\bf Instance $\set{I}_1$:}
Suppose $V = 1$, $\D = 2$, $T=2$, and $g(\cdot)$ is the geometric distribution with parameter $q$, e.g. $g(\tau) = q(1-q)^{\tau-1}$. The arrival probabilities are given by $\lambda_{1, 1} = 1$ and $\lambda_{2, 2} = \frac{\epsilon}{1-q}$, where $\epsilon << 1-q$. The volunteer match probabilities are given by $p_{1, 1} = \epsilon$ and $p_{1,2} = 1$. 
The top left panel of Figure \ref{fig:exhardness} visualizes instance $\set{I}_1$. The following lemma---which we prove in Appendix \ref{proof:exprophet}---states that no online policy can complete more than a $\frac{1}{2-q}$ fraction of $\mathbf{LP}_{\set{I}_1}$.

\begin{lemma}[Upper Bound for Instance $\set{I}_1$]
\label{lem:exprophet}
In instance $\set{I}_1$, the expected number of completed tasks under any online policy is at most $\frac{1}{2-q} LP_{\set{I}_1}$ for $q \in [0,1)$.
\end{lemma}

Before proceeding to the second instance, we make two remarks: (i) If $q = 0$, the above instance is equivalent to the {canonical} instance used in the prophet inequality to establish an upper bound of $1/2$ (see, e.g., \citealt{hill1992survey}). (ii) The term $(1-1/e)$ in the competitive ratio of {\sdnpcolor our policy} corresponds to the gap between $f(\mathbf{x^*})$ (defined in \eqref{eq:fdef1}) and the benchmark $\mathbf{LP}$, whereas the $\frac{1}{2-q}$ corresponds to the gap between the performance of our online policy and $f(\mathbf{x^*})$ due to the loss in the online phase. In instance $\set{I}_1$, there is only one volunteer and consequently $f(\mathbf{x^*}) = \mathbf{LP}_{\set{I}_1}$. Therefore, instance $\set{I}_1$  shows that the lower bound achieved in  the online phase of our {\sdnpcolor policy} is tight.

The construction of our next instances are more delicate, as we aim to find instances for which both the loss in the offline phase (i.e., the gap between $f(\mathbf{x^*})$ and  $\mathbf{LP}_{\set{I}}$) and the loss in the online phase (i.e., the gap between the performance of the online policy and  $f(\mathbf{x^*})$) are large. 

\medskip 
{\bf Instance $\set{I}_2$:} Suppose $V = \frac{1}{q} = n$, $\D = 1$, $T=n^2+1$, and $g(\cdot)$ is the geometric distribution with parameter $q$, e.g. $g(\tau) = q(1-q)^{\tau-1}$. The arrival probabilities are given by $\lambda_{1, 1} = 1$ and $\lambda_{1, t} = q$ for $t \in [T] \setminus [1]$. The volunteers are homogeneous with $p_{v, 1} = q$ for all $v \in [V]$. The top right panel of Figure \ref{fig:exhardness} visualizes instance $\set{I}_2$. The following lemma---which is proven in Appendix \ref{proof:extightsmallq}---states that no online policy can complete more than a $1 + q - \frac{q(1-q)}{\log(\frac{1}{1-q})(1+q)}(1-e^{-1})$ fraction of $LP_{\set{I}_2}$.

\begin{lemma}[Upper Bound for Instance $\set{I}_2$]
\label{lem:extightsmallq}
In instance $\set{I}_2$, the expected number of completed tasks under any online policy is at most $\Big[1 + q - \frac{q(1-q)}{\log(\frac{1}{1-q})(1+q)}(1-e^{-1})\Big] \mathbf{LP}_{\set{I}_2}$ , where $q \in \{1/n,  n \in \mathbb{N}$\}.
\end{lemma}

The proof of this lemma involves three steps: (i) lower-bounding $\mathbf{LP}_{\set{I}_2}$ by finding a feasible solution, (ii) establishing that always notifying every volunteer is the best online policy, and (iii) assessing the performance of this policy relative to $\mathbf{LP}_{\set{I}_2}$. A full proof can be found in Appendix \ref{proof:extightsmallq}.

\medskip 
{\revcolor {\bf Instance $\set{I}_3$:} Suppose $V = \bigvalvol$ for sufficiently large $\bigvalvol$, $\D = 1$, $T=\bigval^2$, and the inter-activity time distribution is deterministic with length $\bigval$, e.g. $g(\tau) = \mathbb{I}(\tau = \bigval)$. 
We emphasize that $q = 0$ for such a distribution.
The arrival probabilities are given by $\lambda_{1, t} = \frac{1}{\bigval}$ for all $t \in [T]$. The volunteers are homogeneous with $p_{v, 1} = \frac{1}{\bigvalvol}$ for all $v \in [V]$. 
The bottom left panel of Figure \ref{fig:exhardness} visualizes instance $\set{I}_3$.
The following lemma---which is proven in Appendix \ref{proof:q0}---states that no online policy can complete more than a $0.334$ fraction of $LP_{\set{I}_3}$.

\begin{lemma}[Upper Bound for Instance $\set{I}_3$]
\label{lem:exq0}
In instance $\set{I}_3$, the expected number of completed tasks under any online policy is at most $0.334 \times \mathbf{LP}_{\set{I}_3}$.
\end{lemma}

Instance $\set{I}_3$ is quite similar to instance $\set{I}_2$, and correspondingly, the proof of Lemma \ref{lem:exq0} builds on ideas in the proof of Lemma \ref{lem:extightsmallq}. A full proof can be found in Appendix \ref{proof:q0}.}

\begin{figure}[t]
\centering
 \includegraphics[trim=65 515 65 70, clip, width=.95\textwidth]{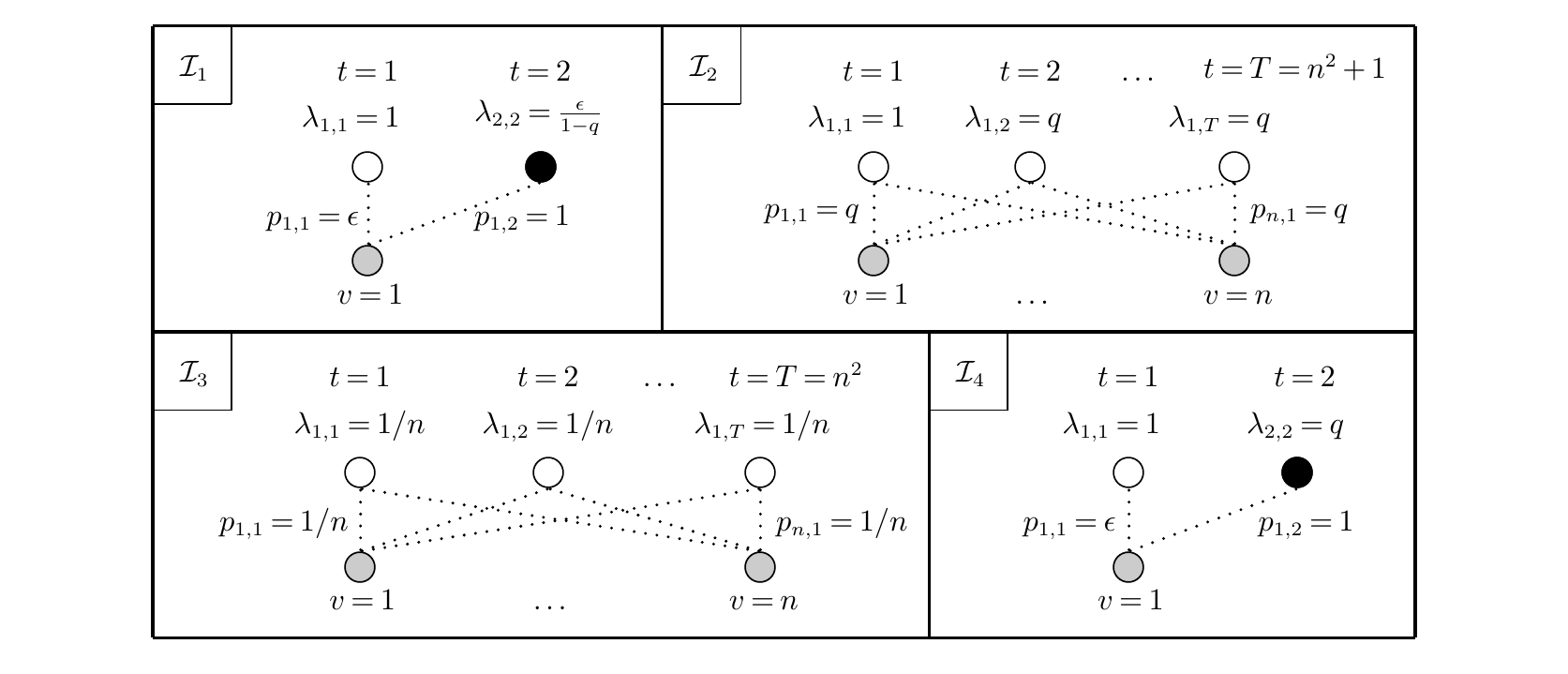}
\caption{Visualizations of instances $\set{I}_1$, $\set{I}_2$, $\set{I}_3$, and $\set{I}_4$.} \label{fig:exhardness}
\end{figure}

\subsection{Upper Bound on Following Ex Ante Solution}
\label{subsec:upperboundexante}
As noted in Section \ref{subsec:offline}, directly following the ex ante solution $\mathbf{x^*}$ (as defined in \eqref{eq:xstar}) does not achieve a good competitive ratio because volunteers are not always active at the ``right time'' if notifications only respect their inactivity period \emph{in expectation}. 
To highlight this intuition, 
we state the following proposition which establishes an upper bound on the performance of {such a} policy.


\begin{proposition}[Upper Bound on Performance of Following Ex Ante Solution]
\label{prop:ubexante}
Suppose the MDHR of the inter-activity time distribution is {$ q > 0$}. Then the policy of directly following the ex ante solution achieves a competitive ratio of at most $q$.
\end{proposition}

\begin{proof}{Proof:}
In order to prove this proposition, we construct an  instance with a single volunteer such that, under the policy of directly following the ex ante solution,  
she  is likely to be inactive when she is most valuable. In particular, we define instance $\set{I}_4$ as follows:

{\bf Instance $\set{I}_4$: } Suppose $V = 1$, $\D = 2$, $T=2$, and $g(\cdot)$ is the geometric distribution with parameter $q$, e.g. $g(\tau) = q(1-q)^{\tau-1}$. The arrival probabilities are given by $\lambda_{1, 1} = 1$ and $\lambda_{2, 2} = q$. The volunteer match probabilities are given by $p_{1, 1} = \epsilon$ and $p_{1,2} = 1$, where $\epsilon << q$. The bottom right panel of Figure \ref{fig:exhardness} visualizes instance $\set{I}_4$. In the rest of the proof, we show that for instance $\set{I}_4$, the expected number of completed tasks when directly following the ex ante solution is $q\mathbf{LP}_{\set{I}_4}$.

We first claim that the ex ante solution in instance $\set{I}_4$ is $x^*_{1,1,1} = 1$ and $x^*_{1, 2, 2} = 1$. {Such} a solution is feasible in (LP), and because the objective of (LP) is monotonically increasing in $x_{1,1,1}, x_{1, 2, 2} \in [0,1]$, no other feasible solution can achieve a higher value. As a consequence, $\mathbf{LP}_{\set{I}_4} = \epsilon + q$.

An online policy of following $\mathbf{x^*}$ (the ex ante solution) is therefore equivalent to always notifying volunteer $1$ when there is an arrival. Such a policy completes a task in the first period with probability $\epsilon$ and a task in the second period with probability $q^2$ (note that she will become active in period 2 with probability $q$ and there will be an arrival of type $2$ independently with probability $q$). The expected number of tasks completed is then given by $q(1+ \frac{\epsilon(1-q)}{q(q+\epsilon)})\mathbf{LP}_{\set{I}_4}$. For $\epsilon<<q$, this corresponds to a competitive ratio of $q$.
\Halmos \end{proof}

We conclude this section by noting that instance $\set{I}_4$ also highlights the challenges that the reactivation of volunteers brings to the design of online policies for our online volunteer  notification problem compared to the classic online matching setting---in which once a resource is matched, it will not return. 
In particular, papers such as \citet{OnlineWeighted} show that for the classic online matching setting, relying on an ex ante solution leads to constant-factor competitive ratios. 
However, instance $\set{I}_4$ shows that this is not the case when volunteers can reactivate. In other words, generalizing online stochastic matching (by allowing for the return of resources) comes with technical challenges that require us to employ new design ideas.

\section{Evaluating Policy Performance on FRUS Data}
\label{sec:data}
In this section, we use data from FRUS to evaluate the performance of the {\sdnpcolor SN policy} described in Section \ref{sec:algos}. First, we briefly explain how we use data to determine the model primitives. Then we exhibit the superior performance of our {\sdnpcolor policy} compared to different benchmarks, including policies that resemble the strategies used at various FRUS locations.

{\bf Estimating model primitives:} As explained in Section \ref{sec:model}, in order to define an instance of the online volunteer notification problem, we must determine the match probabilities, i.e., $\{p_{v, \don}: v \in [V], \don \in [\D]\}$; the arrival rates of tasks, i.e., $\{\lambda_{\don, t}: \don \in [\D], t \in [T] \}$; and the inter-activity time distribution $g(\cdot)$. {\revcolor For each location, we considered a six-week horizon along with a moderate number of volunteers ($V \in [15, 20]$) and a reasonable granularity of task types ($\D \in [30, 75])$.}
\begin{itemize}
\item[] {\bf Match probabilities:} As evidenced in Figure \ref{fig:pcatotal}, volunteer preferences over task types are heterogeneous and predictable. To come up with estimates $\{\hat{p}_{v,\don}: v \in [V], \don \in [\D] \}$ for each FRUS location, we first create a feature vector for each task. We then build a $k$-Nearest Neighbors classification model, tuning the parameter $k$ using cross-validation. The AUC values of such classification models range between {0.85 and 0.95} across tested locations. 
\item[] {\bf Arrival Rates:} Recall that for FRUS, a task is a food rescue (donation) that remains available on the day of delivery. We define a set of task types based on pick-up and drop-off location, day-of-week, and time-of-day. Most food rescues repeat on a weekly cycle, and using a logistic regression model based on type and historical sign-up patterns, we can make accurate predictions about whether such rescues will be available on the day of delivery. The AUC values of these logistic regression models range between 0.78 and 0.96 across tested regions when making predictions over a six-week horizon.\footnote{In Appendix \ref{app:robustness}, we present numerical results that demonstrate the robustness of {\sdnpcolor our policy}.}
\item[] {\bf Inter-activity time distribution:} At FRUS, many site directors follow a policy of waiting at least a week before notifying the same volunteer about another last-minute food rescue. Many other volunteer notification systems operate under similar self-imposed restrictions on notification frequency, such as notifications about blood donation studied in \citet{dickerson2019blood}. Consequently, for the results reported in Figure \ref{fig:numerics1}, we assume the inter-activity time is deterministic and equal to seven days.
\end{itemize}

\begin{figure}[t]
\centering
\scalebox{1.0}[1.0]{
 \includegraphics[width=.94\textwidth]{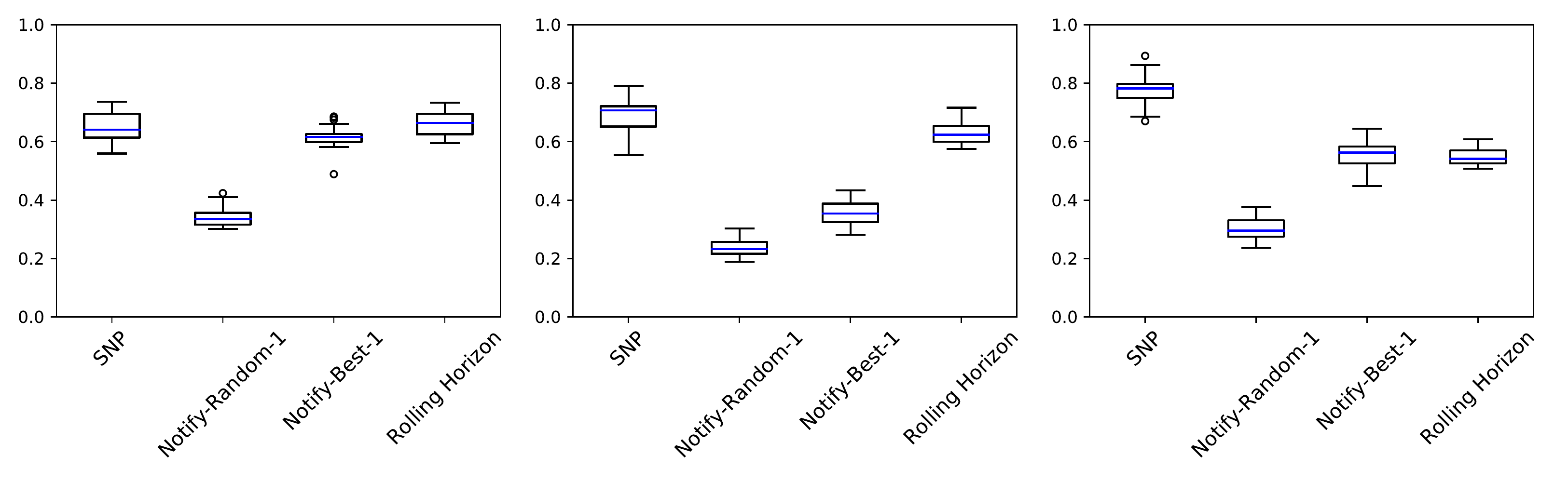}}
\caption{The fraction of $\mathbf{LP}$ achieved in Locations (a), (b), and (c) (left, middle, and right, respectively) by the SN policy and heuristics assuming a deterministic inter-activity time.}
\label{fig:numerics1}
\end{figure}

In the following, we compare the performance of our {\sdnpcolor online policy} to several benchmarks, {including a strategy that simulates the current practice at various FRUS locations.} We use instances constructed with data from three different locations as described above. Whenever a task arrives, FRUS site directors tend to notify a small subset of volunteers, {chosen haphazardly among `eligible' volunteers. We simulate this common practice with a `notify-random-$n$' heuristic that notifies $n$ volunteers chosen uniformly at random among eligible volunteers. Because site directors usually follow the practice of waiting at least a week before notifying the same volunteer again, a volunteer is considered eligible if she has not been notified for at least 6 days. We remark that the subset of such eligible volunteers perfectly coincides with the subset of active volunteers when the inter-activity time distribution is deterministically equal to one week.
This highlights that our framework perfectly captures the current practice. 
In addition to the commonly-used heuristic described above, we
compare our {\sdnpcolor policy} to two stronger, data-driven heuristics. First, we consider a `notify-best-$n$' heuristic that greedily notifies the $n$ eligible volunteers with the largest matching probabilities, i.e., the highest values of $p_{v, \don}$. We also simulate an adaptive `rolling horizon' heuristic that solves a one-week version of (LP) for eligible volunteers whenever a task arrives and probabilistically follows that solution.}

Figure \ref{fig:numerics1} displays the ratio between each policy and $\mathbf{LP}_\set{I}$ across 25 simulations. We highlight that the SN policy displays remarkable consistency, {significantly outperforming the current FRUS practice which is represented by the `notify-random-$n$' heuristic. In addition, the SN policy significantly outperforms the two stronger heuristics in at least one location (e.g., Location (c)) while always performing at least as well.}\footnote{We remark that  for `notify-random-$n$' and `notify-best-$n$', we have optimized over $n$, which turned out to be $1$
in all 3 regions and for both heuristics.} Further note that the SN policy's performance significantly exceeds its competitive ratio {of $\frac{1}{2}(1-\frac{1}{e})$}, as given in Theorem \ref{thm:alg2}.

Moving beyond the common practice at FRUS---which coincides with our setting when the inter-activity time distribution is deterministic---we note that 
our {\sdnpcolor policy provides broader guarantees} for a general inter-activity time distribution. As such, we also study {\sdnpcolor its} numerical performance when faced with stochastic inter-activity times. One natural choice to study is the geometric distribution, which is assumed for the model of volunteer inter-activity time in previous work such as \citet{ata2016dynamic}. To that end, in Figure \ref{fig:numerics2}, we provide numerical results on FRUS instances, but when modifying the inter-activity time distribution and assuming it is geometric with mean of 7 days. {\revcolor We emphasize that all other problem parameters remain the same.}
Figure \ref{fig:numerics2} displays the ratio between each policy and $\mathbf{LP}_\set{I}$ across 25 simulations for such a setting. 
Since the platform does not observe whether or not a volunteer is active, we compare our {\sdnpcolor policy} to other benchmarks suitable for this setting. First, we consider a simplistic `notify-all' heuristic that resembles a mass notification system commonly used in nonprofit settings. Next, we consider a highly intelligent heuristic `notify-upto-$\rho$,'
 which notifies volunteers in descending order of $p_{v,\don}$ until the probability of any volunteer responding positively exceeds $\rho$. {As an example, suppose a task of type $\don$ arrives at time $t$, with $p_{1,\don} \geq p_{2,\don} \geq \ldots \geq p_{V,\don}$; {\revcolor in addition, let $a_v$ for all $v \in [V]$ denote the probability that volunteer $v$ is currently active  when employing the `notify-upto-$\rho$' heuristic. In this example, under such a heuristic} the platform notifies volunteers $1, 2, \ldots, u$ where $u$ marks the volunteer for which $1-\prod_{v=1}^{u}(1- p_{v,\don}a_{v})  \geq \rho$ but $1-\prod_{v=1}^{u-1}(1- p_{v,\don}a_{v})  < \rho$.} 
We consider two values of  $\rho$, $0.25$ and $0.5$.
{Once again, the SN policy vastly outperforms the simple heuristic (in this case, `notify-all') in all three locations. Further, it exhibits significantly better performance than each heuristic in at least one location and always performs comparably to the best heuristic. }

Given that all of the model primitives except the inter-activity time distribution are the same  for the results presented in Figures \ref{fig:numerics1} and \ref{fig:numerics2}, 
we can assess how the performance of our {\sdnpcolor policy changes when the inter-activity time  changes from being deterministic to being random and geometrically distributed (with the same mean). 
Averaging across the three locations,
the total number of tasks completed by the SN policy is on average 24\% more in the former, i.e., deterministic, setting than in the latter, i.e., geometric. 
However, the benchmark is also larger when the inter-activity time is deterministic by an average of 5\%.
As a consequence, in the former setting, the SN policy achieves ratios that are 19\% better than in the latter setting. We remark that our policy performs better in the deterministic setting despite having a better worst-case guarantee in the geometric setting (recall that the worst-case guarantee is increasing in the MDHR, $q$, {which is $0$ in the deterministic setting}).}

\begin{figure}[t]
\centering
\scalebox{1.0}[1.0]{
 \includegraphics[width=.98\textwidth]{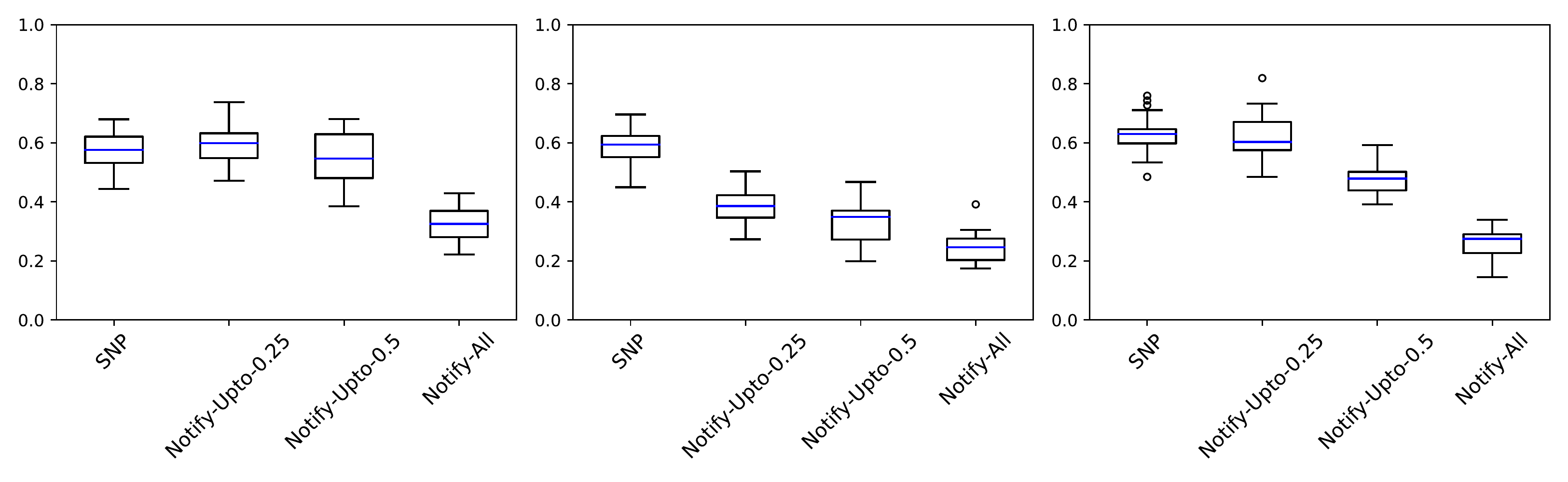}}
\caption{The fraction of $\mathbf{LP}$ achieved in Locations (a), (b), and (c) (left, middle, and right, respectively) by the SN policy and heuristics assuming a geometric inter-activity time.}
\label{fig:numerics2}
\end{figure}

\section{Conclusion}
\label{sec:conclude}

In this paper, we take an algorithmic approach to a commonly faced challenge on volunteer-based crowdsourcing platforms: how to utilize volunteers for time-sensitive tasks at the ``right'' pace while {maximizing the number of completed} tasks. We introduce the online volunteer notification problem {to model} volunteer behavior as well as the trade-off that the platform faces in this online decision making process. 
We develop {\sdnpcolor an online policy that achieves a constant-factor guarantee} parameterized by the MDHR of the volunteer inter-activity time distribution, which gives insight into the impact of volunteers' activity level. The guarantees provided by our {\sdnpcolor policy}  are close to the upper bound we establish for the performance of any online policy.

{\revcolor
Beyond volunteer crowdsourcing, our general framework can be used to design targeted notification systems for other purposes, such as marketing. (As mentioned in the introduction, similar negative reactions to excessive notifications have been documented in that literature.) 
{To address practical considerations} of broader applications, here we discuss how our base model can be readily extended in three dimensions {(for simplicity, we retain the terminology of volunteers and tasks)}.
(i) If tasks generate different rewards, we can {incorporate task weights into our model}. Then, we can simply adjust our SN policy to account for the task weights, which attains the same theoretical guarantees as in the unweighted case presented.
(ii) If volunteers differ in terms of their reaction to notifications, we can adjust our model to accommodate heterogeneous inter-activity time distributions  (our theoretical results extend as long as $q$ is taken to be the infimum of all the heterogeneous MDHRs).
(iii) If a volunteer's preferences exhibit seasonality,
we can incorporate time-varying compatibilities between volunteers and tasks (i.e., $p_{v, \don, t}$ instead of $p_{v, \don}$) without impacting our upper and lower bounds, which allows for settings where volunteers are not initially active.}

{\revcolor We also discuss a few ways of further extending our framework and results which could be valuable directions for future research.} If more than one task can arrive in a given period, our framework could be adjusted to allow the platform to present volunteers with a subset of available tasks when sending a notification.  Our {\sdnpcolor SN policy} could be implemented in such a setting by incorporating volunteers' choice---when faced with a subset---through individualized discrete choice functions. 
Analyzing the performance of our policy in such a setting is an interesting research direction. {\revcolor In the same vein, if some tasks are not prohibitively time-sensitive, the platform could consider batching tasks to improve match efficiency. While our lower bound (achieved by the SN policy) still holds for such a setting, designing policies that outperform SN would be a fruitful direction.
Expanding our framework to include learning would be valuable in settings where the volunteer pool rapidly expands and exploration is needed to ascertain preferences. Additional empirical study of volunteer behavior could shed light on whether the inter-activity time distribution depends significantly on volunteers' responses. From an algorithmic perspective, incorporating a response-dependent inter-activity time distribution would require new technical ideas, since a volunteer's inter-activity time after being notified for a task would depend on the subset of other volunteers notified for the same task.}
Lastly, 
in this work, we measure the performance of an online policy by comparing it to an LP-based benchmark which upper-bounds a clairvoyant solution. From a theoretical perspective, considering other benchmarks (perhaps less strong) is an interesting future direction. 

This work is motivated by our collaboration with FRUS,  a leading volunteer-based food recovery platform,  analysis of whose data confirms that, by and large, volunteers have persistent preferences. Leveraging historical data, we estimate the match probability between volunteer-task pairs as well as the arrival rate of tasks. This enables us to test our {\sdnpcolor policy} on FRUS data from different locations and illustrate {\sdnpcolor its} effectiveness compared to common practice. 
From an applied  perspective,  developing  decision tools that can be integrated with the FRUS app is an immediate next step that we plan to pursue. Finding other platforms that can benefit from our work is another direction for future work.

\ACKNOWLEDGMENT{The authors gratefully acknowledge the Simons Institute for the Theory of Computing, as this work was done in part while attending the program on Online and Matching-Based Market Design.}

{ 
\bibliographystyle{plainnat}
\bibliography{references}
}
\newpage

\begin{APPENDICES}


\section{Proofs for Section \ref{sec:model}}
\subsection{Proof of Proposition \ref{prop:LP}}
\label{proof:LP}
To show that $\mathbf{LP}$ is an upper bound on the clairvoyant solution, we will construct a feasible solution $\mathbf{x} \in \mathcal{P}$ based on the clairvoyant solution. We will then prove that the value of this solution is an upper bound on the value of the clairvoyant solution. 

Let us define the random variables representing inter-activity times as $\vec{Z} \in \vec{\mathcal{Z}} = \mathbb{N}^{V \times T}$, where $Z_{v,t}$ is the inter-activity time of volunteer $v$ if notified at time $t$. In addition, we denote the random arrival sequence as $\vec{\D} \in \vec{\mathcal{\D}} =  [\D]^{T}$, where $\D_t$ is the arrival at time $t$. Finally, suppose we have an indicator variable $\omega_{v,t}(\vec{\don}, \vec{z})$, which is equal to one if and only if the clairvoyant solution contacts volunteer $v$ at time $t$ when the arrival order is given by $\vec{\don}$ and the inter-activity times are given by $\vec{z}$. Recall that we consider a clairvoyant solution that knows the arrival order and the status of each volunteer, but it does not know $Z_{v,t}$ until after time $t$.
Consequently, $\omega_{v,t}(\vec{\don}, \vec{z})$ cannot depend on $z_{v,t'}$ for $t' \geq t$.  

For any volunteer $v$, task of type $j$, and time $t$, we define $\mathbf{\hat{x}} := \{\hat{x}_{v, j, t} : v \in [V], j \in [\D], t \in [T]\}$ such that $$\hat{x}_{v, j, t} = \sum_{\vec{\don} \in \vec{\mathcal{\D}}}\sum_{\vec{z} \in \vec{\mathcal{Z}}} \P{\vec{\D} = \vec{\don}|\D_t = j} \P{\vec{Z} = \vec{z}} \omega_{v,t}(\vec{\don}, \vec{z}).$$ 

To show that $\mathbf{\hat{x}} \in \mathcal{P}$ (see Definition \ref{def:P}), we immediately note that $\hat{x}_{v,j,t} \in [0,1]$, since we are summing indicator variables over a probability distribution. We now need to show that constraint $\eqref{eq:lpcon2}$ is met, namely that $1 \geq  \sum_{t' =1}^t \sum_{j =   1}^{\D} \lambda_{j,t} \hat{x}_{v,j,t} (1-G(t-t'))$. Note that for a given sequence of arrivals $\vec{\don}$ and inter-activity times given by $\vec{z}$, we must have
\begin{align}
1 &\geq \sum_{t'=1}^t  \omega_{v,t'}(\vec{\don},\vec{z}) \I{z_{v,t'} > t-t'} 
\end{align}
This is because both $\omega_{v,t'}(\vec{\don},\vec{z})$ and $\I{z_{v,t'} > t-t'}$ are indicator variables, and if both equal $1$ at time $t'$, then the volunteer $v$ must be inactive until after time $t$. Since the clairvoyant solution only notifies active volunteers, if volunteer $v$ is inactive from $t'$ until after $t$, then $\omega_{v,t''}(\vec{\don},\vec{z}) = 0$ for all $t'' \in [t'+1, t]$. Thus, the sum from $t'=1$ to $t'= t$ of the product of these indicator variables cannot exceed $1$. We now take a weighted sum over all possible arrival sequences and inter-activity times:
\begin{align}
1 \geq& \sum_{t'=1}^t  \sum_{\vec{\don} \in \vec{\mathcal{\D}}} \P{\vec{\D} = \vec{\don}}\sum_{\vec{z} \in \vec{\mathcal{Z}}}\P{\vec{Z} = \vec{z}}  \omega_{v,t'}(\vec{\don},\vec{z}) \I{z_{v,t'} > t-t'} \nonumber \\
=&\sum_{t'=1}^t  \sum_{\vec{\don} \in \vec{\mathcal{\D}}} \P{\vec{\D} = \vec{\don}}\left(\sum_{\vec{z} \in \vec{\mathcal{Z}}}\P{\vec{Z} = \vec{z}}  \omega_{v,t'}(\vec{\don},\vec{z})\right) \left(\sum_{\vec{z} \in \vec{\mathcal{Z}}}\P{\vec{Z} = \vec{z}} \I{z_{v,t'} > t-t'} \right) \label{eq:lpproof1} \\
=&\sum_{t'=1}^t  \sum_{\vec{\don} \in \vec{\mathcal{\D}}} \P{\vec{\D} = \vec{\don}}\left(\sum_{\vec{z} \in \vec{\mathcal{Z}}}\P{\vec{Z} = \vec{z}}  \omega_{v,t'}(\vec{\don},\vec{z})\right) (1-G(t-t')) \label{eq:lpproof2} \\
=&\sum_{t'=1}^t  \sum_{j=1}^{\D} \lambda_{j,t'} \sum_{\vec{\don} \in \vec{\mathcal{\D}}} \P{\vec{\D} = \vec{\don} | \D_{t'} = j}\left(\sum_{\vec{z} \in \vec{\mathcal{Z}}}\P{\vec{Z} = \vec{z}}  \omega_{v,t'}(\vec{\don},\vec{z})\right) (1-G(t-t')) \label{eq:lpproof3} \\
=& \sum_{t'=1}^t \sum_{j=1}^{\D} \lambda_{j,t'}\hat{x}_{v,j,t'}(1-G(t-t')) \label{eq:lpproof4}
\end{align}
In line \eqref{eq:lpproof1}, we use the independence of $\omega_{v,t'}(\vec{\don},\vec{z})$ and $\I{z_{v,t'} > t-t'}$ to rewrite the expected value of their product as the product of their expectations. We substitute in the expected value of $\I{z_{v,t'} > t-t'}$ in line \eqref{eq:lpproof2}. In line \eqref{eq:lpproof3}, we use the law of total probability to sum over all possible arriving task types in time $t'$. We then substitute in the definition of $\mathbf{\hat{x}}$ in line \eqref{eq:lpproof4}. This proves that $\mathbf{\hat{x}} \in \set{P}$. 

It remains to be shown that $\sum_{t=1}^T \sum_{\don=1}^{\D}   \lambda_{\don,t} \min\{ \sum_{v =1}^{V} \hat{x}_{v,\don,t}p_{v,\don}, 1\}$ is at least the value of the clairvoyant solution. Let $\mathcal{C}_{\don, t}$ be the event that a task of type $\don$ arrives at time $t$ and is completed when following the clairvoyant solution. We must have $\P{ \mathcal{C}_{\don, t}} \leq \lambda_{\don,t}$. In addition, since a volunteer must respond in order to complete a task, we must have $$\P{\mathcal{C}_{\don, t}} \leq \lambda_{\don,t}\sum_{\vec{\don} \in \vec{\mathcal{\D}}}\sum_{\vec{z} \in \vec{\mathcal{Z}}} \P{\vec{\D} = \vec{\don}|\D_t = \don} \P{\vec{Z} = \vec{z}} \sum_{v =1}^{V} \omega_{v,t}(\vec{\don}, \vec{z})p_{v,\don} = \lambda_{\don,t} \sum_{v =1}^{V} \hat{x}_{v, \don, t}p_{v,\don}.$$ Combining these two bounds and summing over all task types and time periods, we see that the clairvoyant solution must be at most $\sum_{t=1}^T \sum_{\don=1}^{\D}   \lambda_{\don,t}  \min\{\sum_{v =1}^{V} \hat{x}_{v,\don,t}p_{v,\don}, 1\}$.
Since $\mathbf{\hat{x}} \in \set{P}$ and achieves a weakly larger value than the clarivoyant solution, we have shown that $\mathbf{LP}$ is an upper bound on the clairvoyant solution.

\section{Proofs for Section \ref{sec:algos}}
\subsection{Proof of Lemma \ref{lem:falg}}
\label{proof:falg}
Let $$\hat{f}(\mathbf{x}) := \sum_{v=1}^V f_v(\mathbf{x}) = \sum_{v=1}^V \sum_{t=1}^T \sum_{\don=1}^{\D} \lambda_{\don,t} \left(\prod_{u < v}(1-p_{u, \don} x_{u, \don, t})\right)  p_{v, \don} x_{v, \don, t}.$$ We prove by induction on $V$ that $f(\mathbf{x}) = \hat{f}(\mathbf{x})$, where $f(\mathbf{x})$ is defined in $\eqref{eq:fdef1}$. As a base case, suppose $V=1$. In this case, $f(\mathbf{x}) = \sum_{t=1}^T \sum_{\don=1}^{\D} \lambda_{\don,t} \left(1 - \prod_{v=1}^{1} (1- p_{v,\don}x_{v,\don,t}) \right) = f_1(\mathbf{x})$ so $f(\mathbf{x})$ and $\hat{f}(\mathbf{x})$ are equivalent.

Now suppose this holds for $V = k$. We will show $f(\mathbf{x}) = \hat{f}(\mathbf{x})$ when $V = k+1$. 
\begin{align}
    f(\mathbf{x}) =& \sum_{t=1}^T \sum_{\don=1}^{\D} \lambda_{\don,t} \left(1 - \prod_{v =1}^{k+1} (1- x_{v,\don,t}p_{v,\don}) \right) \\
    =& \sum_{t=1}^T \sum_{\don=1}^{\D} \lambda_{\don,t} \left(1 - (1-x_{k+1,\don,t}p_{k+1,\don})\prod_{v =1}^{k} (1- x_{v,\don,t}p_{v,\don}) \right) \\
    =&\sum_{t=1}^T \sum_{\don=1}^{\D} \lambda_{\don,t} \left(1 - \prod_{v =1}^{k} (1- x_{v,\don,t}p_{v,\don}) \right) \nonumber \\
    &+ \sum_{t=1}^T \sum_{\don=1}^{\D} \lambda_{\don,t} x_{k+1,\don,t}p_{k+1,\don}\prod_{v =1}^{k} (1- x_{v,\don,t}p_{v,\don}) \\
    =&\sum_{v =1}^k f_v(\mathbf{x}) + \sum_{t=1}^T \sum_{\don=1}^{\D} \lambda_{\don,t} x_{k+1,\don,t}p_{k+1,\don}\prod_{v =1}^{k} (1- x_{v,\don,t}p_{v,\don}) \label{eq:proof:falg1}\\
    =&\sum_{v =1}^{k+1} f_v(\mathbf{x}) \\
    =&\hat{f}(\mathbf{x})
\end{align}
All steps are algebraic except for Line \eqref{eq:proof:falg1}, which makes use of the inductive hypothesis. This completes the proof by induction that the two formulas are algebraically equivalent.

\subsection{Proof of Proposition \ref{prop:fxstar}}
\label{proof:fxstar}
To prove this proposition, we first focus on a particular task type $\don$ at a particular time $t$ and prove that for any $\mathbf{x} \in \set{P}$,
$$1-\prod_{v=1}^{V}(1-x_{v,\don,t}p_{v,\don}) \geq (1-\frac{1}{e})\min \{\sum_{v =1}^{V}  x_{v,\don,t}p_{v,\don}, 1\}.$$ 

  To prove the above inequality, we find the minimum possible value of $1-\prod_{v=1}^{V}(1-x_{v,\don,t}p_{v,\don})$ when $\min \{\sum_{v =1}^{V}  x_{v,\don,t}p_{v,\don}, 1\}$ is fixed and equal to $c \in [0,1]$. This is equivalent to solving the program: 
\begin{align}
    \text{maximize}_{\{x_{v, \don , t}: v \in [V]\}} \quad \  &\prod_{v=1}^{V}(1-x_{v,\don,t}p_{v,\don}) \label{eq:fxstarproof} \\
    \text{subject to} \quad \quad \quad  c=& \sum_{v =1}^{V} x_{v,\don,t}p_{v,\don} \nonumber
\end{align}
\begin{claim}
\label{claim:prooflemfxstar1}
For a fixed number of volunteers $V = n$, the value of \eqref{eq:fxstarproof} is less than or equal to  $(1-\frac{c}{n})^n$.
\end{claim}
\begin{proof}{Proof:}
First, we make a change of variables $y_{v,\don,t} = x_{v,\don,t}p_{v,\don}$, where $y_{v,\don,t} \in [0, p_{v,\don}]$. Relaxing this constraint to $y_{v,\don, t} \in [0, 1]$ provides an upper bound on \eqref{eq:fxstarproof}. We now prove by induction on $n$ that the solution to this relaxed problem is $y_{v,\don,t} = \frac{c}{n}$ for all $v$. In the base case with one volunteer, the objective becomes maximizing $y_{1,\don,t}$ subject to $y_{1,\don,t} = c$, which has a trivial solution.

We now assume this holds for $n=k$. If there are $k+1$ volunteers, we consider the problem 
\begin{align}
    \text{maximize}_{\{y_{v, \don , t}: v \in [k+1]\}} \quad \  &(1-y_{k+1,\don, t})\left(\prod_{v \leq k}(1-y_{v,\don,t})\right) \nonumber \\
    \text{subject to} \quad \quad \quad  &c= y_{k+1,\don,t} + \sum_{v=1}^k y_{v,\don,t} \nonumber
\end{align}
For any given $y_{k+1, \don, t} \in [0,1]$, we can apply the inductive hypothesis to solve $y_{v,\don,t} = \frac{c-y_{k+1, \don, t}}{k}$ for $v \in [k]$. This yields a single variable maximization problem with objective function $(1-y_{k+1,\don,t})(1-\frac{c-y_{k+1, \don, t}}{k})^k$. Taking the derivative, this has first order condition of $$(1-\frac{c-y_{k+1, \don, t}}{k})^{k-1}(1-y_{k+1,\don,t}-1+\frac{c-y_{k+1, \don, t}}{k}) = 0.$$ 
One can verify that the solution $y_{k+1,\don,t} = \frac{c}{k+1}$ is the maximum. This implies that $y_{v,\don,t} = \frac{c-y_{k+1, \don, t}}{k} = \frac{c}{k+1}$ for all $v \in [k+1]$, which completes the proof by induction.
Plugging these values for $y_{v,\don,t}$ into the objective function, we get a value of $(1-\frac{c}{n})^n$, which completes the proof of Claim \ref{claim:prooflemfxstar1}.
\Halmos
\end{proof}
Based on this claim, $1-\prod_{v=1}^{V}(1-x_{v,\don,t}p_{v,\don})$ must be at least $1-(1-\frac{c}{n})^n$ when $\sum_{v =1}^{V}x_{v,\don,t}p_{v,\don} = c$. This means that the ratio between the two must be at least $\frac{1-(1-\frac{c}{n})^n}{c}$.
\begin{claim}
For any $n \in \mathbb{N}$ and $c \in [0,1]$, the function $\frac{1-(1-\frac{c}{n})^n}{c}$ is greater than $1-\frac{1}{e}$.
\end{claim}
\begin{proof}{Proof:}
We first show that $n \log(1-\frac{c}{n})$ is increasing in $n$, which implies that the numerator is decreasing in $n$. The derivative of that expression with respect to $n$ is given by $\log(1-\frac{c}{n}) + \frac{c}{n-c} \geq \frac{-c}{n-c} + \frac{c}{n-c} = 0$, using the inequality $\log(1-x) \geq \frac{-x}{1-x}$. This means that the function is decreasing in $n$, regardless of $c$.

Taking the limit as $n$ gets large, the ratio can be written as $\frac{1-e^{-c}}{c}$. The derivative of this expression with respect to $c$ is given by $\frac{(1+c)e^{-c} - 1}{c^2} = \frac{e^{-c}(1+c - e^c)}{c^2} \leq 0$. Thus, this expression is decreasing in $c$. It is minimized at $c=1$, where it attains a value of $1-\frac{1}{e}$.
\Halmos
\end{proof}
Proving this claim establishes that for any $\mathbf{x} \in \set{P}$, any task of type $\don \in [\D]$, and any time $t \in [T]$, $1-\prod_{v=1}^{V}(1-x_{v,\don,t}p_{v,\don}) \geq (1-\frac{1}{e}) \min\{\sum_{v =1}^{V}  x_{v,\don,t}p_{v,\don}, 1 \}$. If we apply this inequality to $\mathbf{x^*_{LP}}$ and take a weighted sum over all task types and time periods, this completes the proof of the proposition, e.g. $f(\mathbf{x^*_{LP}}) \geq (1-\frac{1}{e})\mathbf{LP}$.

 \section{Proofs for Section \ref{sec:hardness}}
 
 \subsection{Proof of Lemma \ref{lem:exprophet}}
 \label{proof:exprophet}
 In instance $\set{I}_1$, a feasible solution to (LP) is $\hat{x}_{1, 2, 2} = 1$ and $\hat{x}_{1, 1, 1} = 1-\epsilon$. Clearly, this solution meets constraint \eqref{eq:lpcon1} as well as constraint \eqref{eq:lpcon2} at time $1$. To verify that it meets constraint \eqref{eq:lpcon2} at time $2$, note that $(1-\epsilon)(1-q)+\frac{\epsilon}{1-q} = 1-(1-\frac{\epsilon(2-q)}{1-q})q$, which is at most $1$ as long as $\epsilon \leq \frac{1-q}{2-q}$. This solution achieves $(1-\epsilon)\epsilon + \frac{\epsilon}{1-q} = \epsilon(\frac{2-q-(1-q)\epsilon}{1-q})$ completed tasks in expectation.

Since it is clearly optimal to notify $v$ in period $2$, an online policy only has one choice to make: whether or not to notify $v$ in period $1$. If notified in period $1$, $v$ will complete $\epsilon + q\frac{\epsilon}{1-q} = \frac{\epsilon}{1-q}$ tasks in expectation. Otherwise, $v$ will complete  $\frac{\epsilon}{1-q}$ tasks in expectation. Thus, no online policy can achieve a value greater than $\frac{\epsilon}{1-q}$. This represents a competitive ratio of no more than $\frac{1}{2-q-(1-q)\epsilon}$, which approaches $\frac{1}{2-q}$ as $\epsilon$ gets small.
 
 \subsection{Proof of Lemma \ref{lem:extightsmallq}}
 \label{proof:extightsmallq}
 We prove Lemma \ref{lem:extightsmallq} in three steps. First, we show that in instance $\set{I}_2$, $\mathbf{LP} \geq n$. Then, we establish that always notifying every volunteer is the best online policy. Finally, we assess the performance of this policy relative to $\mathbf{LP}$. For ease of reference, we repeat the definition of instance $\set{I}_2$ below:
 \medskip
 
 {\bf Instance $\set{I}_2$:} Suppose $V = \frac{1}{q} = n$, $\D = 1$, $T=n^2+1$, and $g(\cdot)$ is the geometric distribution with parameter $q$, e.g. $g(\tau) = q(1-q)^{\tau-1}$. The arrival probabilities are given by $\lambda_{1, 1} = 1$ and $\lambda_{1, t} = q$ for $t \in [T] \setminus [1]$. The volunteers are homogeneous with $p_{v, 1} = q$ for all $v \in [V]$.

\begin{claim}
In instance $\set{I}_2$, $\mathbf{LP} \geq n$
\end{claim} 
\begin{proof}{Proof:}
To prove this claim, we first show that solution $\hat{x}_{v, 1, t} = 1$ for all $v \in [n]$ and $t \in [T]$ is feasible. Clearly, it satisfies constraint \eqref{eq:lpcon1}. Now consider constraint \eqref{eq:lpcon2} for an arbitrary $v \in [n]$ and $t \in [T]$:
\begin{align}
\sum_{\tau =1}^{t} \lambda_{1, \tau} \hat{x}_{v,1,\tau} (1-q)^{t-\tau} &= (1-q)^{t-1} + \sum_{\tau = 2}^t q(1-q)^{t-\tau} \label{eq:app5proofl1} \\
&= (1-q)^{t-1} + q \sum_{\tau' = 0}^{t-2} (1-q)^{\tau'} \nonumber \\
&= (1-q)^{t-1} + q \frac{1 - (1-q)^{t-1}}{q} \nonumber \\ &= 1 \nonumber
\end{align}

Line \eqref{eq:app5proofl1} comes from plugging in $\lambda_{1,1} = 1$ and $\lambda_{1,\tau} = q$ for $\tau \in [T]\setminus [1]$. Now that we have established the feasibility of $\mathbf{\hat{x}}$, we can calculate the value of (LP) at that solution, which is given by $$\sum_{t=1}^{T} \lambda_{1,t}  \min \left\{\sum_{v =1}^{V} \hat{x}_{v,1,t} q,1\right\}  = \sum_{t=1}^{T} \lambda_{1,t} = 1 + n \geq n.$$
\Halmos
\end{proof}

Now that we have a lower bound on $\mathbf{LP}$, we turn our attention to placing an upper bound on any online policy. We do so with the following claim.

\begin{claim}
Notifying every active volunteer whenever a task of type $1$ arrives completes at least as many tasks in expectation as any other online policy.
\end{claim} 
\begin{proof}{Proof:}
We first note that the best online policy cannot do better in expectation than the best online policy that also knows the status of each volunteer $v \in [V]$ at each time $t \in [T]$ because designing a policy without using that additional information is always an option. Thus, it is sufficient to show that notifying every active volunteer whenever a task of type $1$ arrives achieves a weakly greater expected value than any online policy that knows each volunteer's status. 

We proceed via total backward induction, with a base case at time $T$. Suppose an arrival occurs at time $T$. If there are $\aone$ active volunteers, notifying $\atwo$ of them achieves a value-to-go of $1-(1-q)^{\atwo}$. This is increasing in $\atwo$, which means that the optimal online policy is to contact all active volunteers.

Let $\Jhat_{\tau+1}$ represent the value-to-go when an arrival occurred in period $\tau$ given an online policy that notifies all active volunteers at $\tau$ and all future arrivals. Note that we can only make this representation because the inter-activity times are geometrically distributed, which implies that volunteers' transitions from inactive to active are memoryless. Thus the expected payoff is the same regardless of the choices made before period $\tau$. By convention, we set $\Jhat_{T+1} = 0$.

Now we make the inductive hypothesis that contacting all active volunteers whenever a task of type $1$ arrives is the best online policy for $t \in [T]\setminus [k]$. We will show that if an arrival occurs at time $k$, an optimal online policy is to contact all active volunteers. The expected payoff of contacting $\atwo$ volunteers when $\aone$ are active is given by
\begin{align}
   \h_{k, \aone}(\atwo) =& \P{\text{task completed at time $k$}} + \E{\text{tasks completed from $k+1$ to $T$}} \\
   =& \P{\text{task completed at time $k$}} + \nonumber \\ &\sum_{\tau = k+1}^T \P{\text{next arrival at $\tau$}}\E{\text{tasks completed from $k+1$ to $T$}|\text{next arrival at $\tau$}} \label{eq:hk1} \\
   =&1 - (1-q)^{\atwo} + \label{eq:hk2} \\
   &\sum_{\tau = k+1}^T q(1-q)^{\tau - k - 1} \left(1 - (1-q)^{\aone-\atwo}(1 - q (1 - (1-q)^{\tau - k}))^{n - \aone +\atwo} \right) + \label{eq:hk3} \\
   &\sum_{\tau = k+1}^T q(1-q)^{\tau - k - 1} \Jhat_{\tau+1} \label{eq:hk4}
\end{align}
In line \eqref{eq:hk1}, we use the law of total probability. Line \eqref{eq:hk2} represents the probability of completing the task at time $k$. Each term in the summation in line \eqref{eq:hk3} represents the probability that the next task arrives at $\tau$ and that task gets completed. To compute this probability, first note that we know that at time $\tau$ there will be $\aone-\atwo$ volunteers who are definitely active. The remaining volunteers will be independently active with probability $1 - (1-q)^{\tau - k}$. Thus, each of these remaining volunteers will respond to a notification with probability $q(1 - (1-q)^{\tau - k})$. Since the inductive hypothesis assumes that the online policy will contact every active volunteer at $\tau > k$, the probability of any volunteer completing the next task (conditional on an arrival at $\tau$) is given by $1 - (1-q)^{\aone-\atwo}(1 - q (1 - (1-q)^{\tau - k}))^{n - \aone +\atwo}$. In line \eqref{eq:hk4}, we add the remaining expected number of completed tasks from $\tau+1$ to $T$ after an arrival in period $\tau$, which does not depend on the choice of $\atwo$ due to the memorylessness of the transitions from inactive to active. 

We now define $\Delta_{k,\aone}(\atwo) = \h_{k,\aone}(\atwo+1) - \h_{k,\aone}(\atwo)$ for $0 \leq \atwo \leq \aone-1$. This is the incremental benefit of notifying one additional active volunteer. We have
\begin{align}
    \Delta_{k,\aone}(\atwo) =& (1-q)^{\atwo} - (1-q)^{\atwo+1} + \sum_{\tau = k+1}^T q(1-q)^{\tau - k - 1} \times \nonumber \\ &[(1-q)^{\aone-\atwo}(1 - q (1 - (1-q)^{\tau - k}))^{n - \aone +\atwo } \nonumber \\&- (1-q)^{\aone-\atwo-1}(1 - q (1 - (1-q)^{\tau - k}))^{n - \aone +\atwo+1} ] \nonumber \\
    =& q(1-q)^{\atwo} + \sum_{\tau = k+1}^T q(1-q)^{\tau - k - 1}(1-\frac{1 - q (1 - (1-q)^{\tau - k})}{1-q}) \times \nonumber \\
    &[(1-q)^{\aone-\atwo}(1 - q (1 - (1-q)^{\tau - k}))^{n - \aone +\atwo }] \label{longeq:2} \\
        =& q(1-q)^{\atwo} - \sum_{\tau = k+1}^T q(1-q)^{\tau - k - 1}(q (1-q)^{\tau - k-1}) \times \nonumber \\
    &[(1-q)^{\aone-\atwo}(1 - q (1 - (1-q)^{\tau - k}))^{n - \aone +\atwo }] \label{longeq:3} \\
    \geq& q(1-q)^{\aone-1} - \sum_{\tau = k+1}^T q^2(1-q)^{2(\tau - k - 1)}[(1-q)(1 - q (1 - (1-q)^{\tau - k}))^{n-1}] \label{longeq:4} \\
    \geq& q(1-q)^{n-1} - \sum_{\tau' = 1}^\infty q^2(1-q)^{2\tau' - 1}(1 - q (1 - (1-q)^{\tau'})^{n-1}\label{longeq:5} \\
    =& q(1-q)^{n-1} \left( 1 - \sum_{\tau' = 1}^\infty q(1-q)^{2\tau' - 1}(1 + q(1-q)^{\tau'-1})^{n-1} \right)\label{longeq:6} \\
     \geq& 0 \nonumber
    \end{align}

In line \eqref{longeq:2} we factor like terms. In line \eqref{longeq:3} we simplify the fraction. In line \eqref{longeq:4} we lower bound the expression by replacing $\atwo$ with $\aone-1$, its maximum value. To see that this is a valid lower bound, note that this decreases the first term and increases each term in the summation, which decreases the overall expression. In line \eqref{longeq:5}, we simplify the bounds of the summation and provide a lower bound (an upper bound on a negative term) by summing all the way to infinity and by setting $\aone = n$, which decreases the first term and otherwise has no impact. We then factor out terms and simplify to get \eqref{longeq:6}. Numerically, we can verify that this term is weakly positive. 

Since $\Delta_{k,\aone}(\atwo) \geq 0$ for all $\atwo \leq \aone-1 \leq n-1$, we have proved the inductive hypothesis that notifying all active volunteers in period $k$ is optimal. This completes the proof of the claim.
\Halmos
\end{proof}

We now provide an upper bound on the value of this optimal online policy.

\begin{claim}The value of notifying every active volunteer whenever a task of type $\don$ arrives is at most $q+\left(1 - \frac{q(1-q)}{\log(\frac{1}{1-q})(1+q)}(1-e^{-1}) \right)$. 
\end{claim}

\begin{proof}{Proof:}
If we notify every volunteer at every arrival, then in the first period, we achieve a payoff of $1-(1-q)^n \leq 1$. Now suppose every volunteer was most recently notified at time $\tau-1$. Recall that the expected reward for the remainder of the time horizon is given by $\Jhat_\tau$. By definition, the value of this policy is given by $1-(1-q)^n +\Jhat_2 \leq 1 +\Jhat_2$. 

Starting with $\Jhat_{T+1} = 0$, we can recursively compute $\Jhat_{\tau}$ with the following update equation:
\begin{align}
\Jhat_{\tau} &= \sum_{t=\tau}^{T} q(1-q)^{t-\tau}\left( 1-(1-q(1-(1-q)^{t-\tau+1}))^n \right) + \sum_{t=\tau}^{T} q(1-q)^{t-\tau} \Jhat_{t+1} \label{eq:app5l2}
\end{align}
In equation \eqref{eq:app5l2}, $q(1-q)^{t-\tau}$ represents the probability that the next arrival of a task of type 1 occurs at time $t$. The probability that such a task gets completed is given by $1-(1-q(1-(1-q)^{t-\tau+1}))^n$. Thus, the first sum is the probability that the next arriving task gets completed. The second sum represents the expected number of tasks completed after that arrival.

We will prove via total backward induction that $\Jhat_\tau \leq q(T+1-\tau)\Z$, where $\Z$ represents the expected probability of completing the next task unconditional on when it arrives, i.e. $$\Z:=\sum_{t=1}^\infty q(1-q)^{t-1}(1-(1-q(1-(1-q)^t))^n).$$
In words, we are upper-bounding the value-to-go after an arrival with the expected number of remaining arrivals $q(T+1-\tau)$ times the (unconditional) expected probability of completing the next task.

Clearly, this holds with equality for $\tau = T+1$. We now assume this is true for $\tau \geq k+1$, and we will try to show that $\Jhat_k \leq q(T+1-k)\Z$. First, in step (i), we will be place a bound on the first summation in Line \eqref{eq:app5l2}. Then, in step (ii), we will place a bound on the second summation. 

\medskip

\noindent \textbf{Step (i):} We begin by noting that the probability of completing a task is increasing in the amount of time since the previous arrival. Mathematically, the function $\rho(t) :=$ $1-(1-q(1-(1-q)^{t-k+1}))^n$ is increasing in $t$. Therefore $\max_{t \in [T] \setminus [k-1]}\{\rho(t)\} < \min_{t \in \mathbb{N} \setminus [T]}\{\rho(t)\}$, i.e., every member of the former set is smaller than every member of the latter. Consequently, any convex combination of the former set is smaller than any  convex combination of the latter. Thus we have:
\begin{align}
    \frac{\sum_{t=k}^{T} q(1-q)^{t-k}\left( 1-(1-q(1-(1-q)^{t-k+1}))^n\right)}{\sum_{t=k}^{T} q(1-q)^{t-k}} \leq \frac{\sum_{t=T+1}^{\infty} q(1-q)^{t-k}\left( 1-(1-q(1-(1-q)^{t-k+1}))^n\right)}{\sum_{t=T+1}^{\infty} q(1-q)^{t-k}}
\nonumber    \end{align}
Now we use the above inequality in combination with the algebraic fact that if $\frac{a}{c} \leq \frac{b}{d}$, then $\frac{a}{c} \leq \frac{a+b}{c+d}$ to yield
\begin{align}\frac{\sum_{t=k}^{T} q(1-q)^{t-k}\left( 1-(1-q(1-(1-q)^{t-k+1}))^n\right)}{\sum_{t=k}^{T} q(1-q)^{t-k}} &\leq \frac{\sum_{t=k}^{\infty} q(1-q)^{t-k}\left( 1-(1-q(1-(1-q)^{t-k+1}))^n\right)}{\sum_{t=k}^{\infty} q(1-q)^{t-k}} \nonumber \\
&= \frac{\Z}{\sum_{t=k}^{\infty} q(1-q)^{t-k}} \nonumber \\&= \Z \nonumber \end{align}
Consequently,
\begin{align}
    \sum_{t=k}^{T} q(1-q)^{t-k}\left( 1-(1-q(1-(1-q)^{t-k+1}))^n\right) &\leq \Z\sum_{t=k}^{T}q(1-q)^{t-k} \nonumber \\
    &= \Z(1-(1-q)^{T-k+1}) \label{ex5bound1}
\end{align}
By putting a bound on the first summation in line \eqref{eq:app5l2}, we have completed step (i).

\medskip

\noindent \textbf{Step (ii):} Having upper-bounded the first summation, we now bound the second summation in line \eqref{eq:app5l2}, i.e., $\sum_{t=k}^{T} q(1-q)^{t-k} \Jhat_{t+1}$. Using the inductive hypothesis, we have
\begin{align}
    \sum_{t=k}^{T} q(1-q)^{t-k}\Jhat_{t+1} &\leq q\sum_{t=k}^{T} q(1-q)^{t-k}(T-t)\Z \nonumber  \\
    &= q\sum_{t'=0}^{T-k} q(1-q)^{t'}(T-k-t')\Z \nonumber  \\
    &= q\left(\sum_{t'=0}^{T-k} q(1-q)^{t'}(T-k)\Z -\sum_{t'=0}^{T-k} q(1-q)^{t'}t'\Z\right) \nonumber \\
    &= q(T-k)(1-(1-q)^{T-k+1})\Z - (1-q)(1-(1-q)^{T-k}-(T-k)q(1-q)^{T-k})\Z \label{ex5bound2helper} \\
    &=q(T-k+1)\Z - (1 - (1-q)^{T-k+1})\Z \label{ex5bound2}
\end{align}

The equality in line \eqref{ex5bound2helper} follows from well-known results in computing the summation of power series. Combining \eqref{ex5bound1} and \eqref{ex5bound2} gives us, as desired,
$\Jhat_k \leq q(T-k+1)\Z$, which completes the proof by induction. All that remains is to bound $\Z$, which we do below, starting from its definition:

\begin{align}
    \Z & = \sum_{t=1}^\infty q(1-q)^{t-1}(1-(1-q(1-(1-q)^t))^n) \nonumber \\
    & = \sum_{t=1}^\infty q(1-q)^{t-1}- \sum_{t=1}^\infty q(1-q)^{t-1}(1-q(1-(1-q)^t))^n \label{eq:zbound1} \\&\leq 1 - \int_{t=1}^\infty q(1-q)^{t-1}(1-q(1-(1-q)^t))^n dt \label{eq:zbound2}
    \\&= 1 + \frac{q}{\log(1-q)(1-q)} \int_{u=q}^1 (1-qu)^n du \label{ex5int2}
    \\&= 1 + \frac{q}{\log(1-q)(1-q)} \left(\frac{(1-q^2)^{n+1} - (1-q)^{n+1}}{nq+q}\right) \label{eq:zbound3}
    \\&= 1 + \frac{q(1-q)^n}{\log(1-q)(nq+q)} \left((1+q)^{n+1} - 1\right) \label{eq:zbound4}
    \\ &= 1 - \frac{q(1-q)}{\log(\frac{1}{1-q})(1+q)} (1-q)^{1/q-1}\left((1+q)^{1/q+1} - 1\right) \label{eq:zbound5}
    \\&\leq 1 - \frac{q(1-q)}{\log(\frac{1}{1-q})(1+q)}(1-e^{-1}) \label{ex5expbounds}
\end{align}

In line \eqref{eq:zbound1}, we split the summation into two parts. In line \eqref{eq:zbound2}, we evaluate the first summation and lower-bound the second sum with an integral. Line \eqref{ex5int2} comes from a substitution of $u = 1-(1-q)^t$ and in line \eqref{eq:zbound3} we evaluate the integral. In line \eqref{eq:zbound4}, we factor out $\frac{(1-q)^{n+1}}{nq+q}$. Subsequently, in line \eqref{eq:zbound5}, we make minor algebraic adjustments and make use of the definition $q = 1/n$. The final inequality in line \eqref{ex5expbounds} comes from bounding two auxiliary functions, as described below.

First, we define $\phi_1(q) : = (1-q)^{1/q-1}$. 
This function is increasing  for $q \in (0,1)$, and $\lim_{q \rightarrow 0} \phi_1(q) = e^{-1}$, which means
$(1-q)^{1/q-1} \geq \frac{1}{e}$. 
This enables us to lower bound $\phi_1(q)$ with $e^{-1}$.

Next, we define $\phi_2(q) := (1+q)^{1/q+1} - 1$. This function is also increasing for $q \in (0,1)$, and $\lim_{q \rightarrow 0} \phi_2(q) = e-1$. This enables us to lower bound $\phi_2(q)$ with $e-1$. The inequality in line \eqref{ex5expbounds} results from multiplying the two bounds provided by $\phi_1(q)$ and $\phi_2(q)$.

This implies that $\Jhat_2 \leq q(T-1)(1 - \frac{q(1-q)}{\log(\frac{1}{1-q})(1+q)}(1-e^{-1}))$. Since $q(T-1) = n$ by definition, and since the total expected number of successful tasks is bounded by $1+\Jhat(2)$, we have proven the claim that the value of notifying every active volunteer whenever a task of type $1$ arrives is at most
 $1 + n\left(1 - \frac{q(1-q)}{\log(\frac{1}{1-q})(1+q)}(1-e^{-1}) \right)$. 
 \Halmos
 \end{proof}
 
Putting all three claims together (e.g. taking the upper bound on the expected number of completed tasks by the best online policy and dividing by a lower bound on $\mathbf{LP}$), we have shown that in $\set{I}_2$, no online algorithm can achieve a competitive ratio of more than $1+q - \frac{q(1-q)}{\log(\frac{1}{1-q})(1+q)}(1-e^{-1})$.

\subsection{Proof of Lemma \ref{lem:exq0}}
 \label{proof:q0}
  {\revcolor
 We prove Lemma \ref{lem:exq0} in two steps. First, we show that in instance $\set{I}_3$, $\mathbf{LP}\geq n$. Then, we identify the optimal online policy in a ``less constrained'' setting, which establishes an upper bound of $0.334 \times \bigval$ on the performance of any online policy in instance $\set{I}_3$. For ease of reference, we repeat the definition of the instance below:

 \medskip
 {\revcolor {\bf Instance $\set{I}_3$:} Suppose $V = \bigvalvol$ for sufficiently large $\bigvalvol$, $\D = 1$, $T=\bigval^2$, and the inter-activity time distribution is deterministic with length $\bigval$, e.g. $g(\tau) = \mathbb{I}(\tau = \bigval)$. 
We emphasize that $q = 0$ for such a distribution.
The arrival probabilities are given by $\lambda_{1, t} = \frac{1}{\bigval}$ for all $t \in [T]$. The volunteers are homogeneous with $p_{v, 1} = \frac{1}{\bigvalvol}$ for all $v \in [V]$.} 
 
 \begin{claim}
In instance $\set{I}_3$, $\mathbf{LP} \geq \bigval$
\end{claim} 
\begin{proof}{Proof:}
To prove this claim, we first show that solution $\hat{x}_{v, 1, t} = 1$ for all $v \in [\bigvalvol]$ and $t \in [T]$ is feasible. Clearly, it satisfies constraint \eqref{eq:lpcon1}. Now consider constraint \eqref{eq:lpcon2} for an arbitrary $v \in [\bigvalvol]$ and $t \in [T]$:
\begin{align}
\sum_{\tau = 1}^{t} \lambda_{1, \tau} \hat{x}_{v,1,\tau}(1-G(t-\tau)) \quad = \sum_{\tau = \max\{1, t-\bigval+1\}}^{t} \lambda_{1, \tau} \hat{x}_{v,1,\tau}  \quad = \sum_{\tau = \max\{1, t-\bigval+1\}}^{t} \frac{1}{\bigval} \hat{x}_{v,1,\tau} \quad \leq \quad 1  \nonumber
\end{align}

Now that we have established the feasibility of $\mathbf{\hat{x}}$, we can calculate the value of (LP) at that solution, which is given by $$\sum_{t=1}^{T} \lambda_{1,t}  \min \left\{\sum_{v =1}^{V} \hat{x}_{v,1,t} \frac{1}{\bigvalvol},1\right\}  = \sum_{t=1}^{T} \lambda_{1,t} = \bigval.$$
\Halmos
\end{proof}

We now proceed to the second step and place an upper bound on the performance of any online policy which only notifies active volunteers.\footnote{\revcolor{Note that with deterministic inter-activity times, a policy knows the state of each volunteer at each time period. Thus, it is without loss of generality to restrict attention to policies that only notify active volunteers.}} To do so, we identify the optimal policy in a hypothetical environment with ``enhanced volunteer activity''.

First, we describe this hypothetical environment. We will use the labels  $\{\tau_1, \tau_2, \dots \}$ to denote a subsequence of time periods such that $\tau_1$ is the first period in which an arrival occurs, and $\tau_i$ is the first period in which an arrival occurs \emph{after} $\tau_{i-1} + \bigval$. At each time $\tau_i + \bigval$, all volunteers that are inactive 
transition to a state which we call \emph{pseudoactive}.
Subsequently, such volunteers will transition from {pseudoactive} to active $\bigval$ periods after they were last notified. 
Said differently, their transition time between inactive to active does not change. However, they may spend some of that time interval in a pseudoactive state.

In this new relaxed environment, both active and pseudoactive volunteers may respond to a notification. Therefore, we can consider policies which notify both active and pseudoactive volunteers. {We remark that any online policy which only notifies active volunteers can be implemented in this relaxed setting by simply choosing to ignore the pseudoactive volunteers.} Consequently, we upper bound the performance of any online policy in the original setting by finding the optimal policy in this relaxed setting.

Observe that in this new setting, all volunteers are available for notification (either active or pseudoactive) at $\tau_{i}$ \emph{regardless} of the history so far. 
{Further, all volunteers will also be available for notification at $\tau_{i} + \bigval$ regardless of the notifications sent between $\tau_i$ and $\tau_i +\bigval$.}
Consequently, {we can ignore both the past and the future when determining} the optimal policy for the interval $\tau_{i} \leq t < \tau_{i} +n$.
{The problem of finding the optimal policy in this $\bigval$-period interval is identical for all $i$.} The setting is equivalent to an $n$-period instance where each volunteer can be notified at most once, a task arrives deterministically in period $1$, and tasks arrive in each subsequent period with probability $\frac{1}{\bigval}$. In the following claim, we establish the optimal policy in such a problem along with its corresponding value. Repeating the solution to this $\bigval$-period problem at each $\tau_i$ is the overall optimal policy in this relaxed setting.\footnote{\revcolor{For completeness, we note that if $\tau_i + \bigval > T$, then the optimal policy between $\tau_i$ and $T$ will be different. This impact is vanishing as $\bigval$ gets large.}}

\begin{claim}
Consider an adjusted version of instance $\set{I}_3$ with only $\bigval$ periods and a deterministic arrival in period $1$. In the limit as $\bigval \rightarrow \infty$, the expected number of tasks completed is at most $2 - \frac{1+2\sqrt{e-1}}{e} < 0.668$, which is achieved by a policy ${\pi}^*$ that notifies a fraction $1-\frac{1}{2}\log(e-1)$ of volunteers in period 1 and notifies the remaining fraction of volunteers upon arrival of a second task.
\end{claim}

\begin{proof}{Proof:}
We will first find the optimal policy within a subclass of policies, and then we will show that this is optimal among all policies.

Consider the class of policies which partitions volunteers into two groups such that one group is notified for the first task (at period $1$) and the other group is notified upon the arrival of a second task. (If additional tasks arrive, they are not completed because all volunteers have been notified). Let $z$ denote the fraction of volunteers notified for the first task. This first task is then completed with probability $1-(1-\frac{1}{\bigvalvol})^{z\bigvalvol} \rightarrow 1 - e^{-z}$. A subsequent task will only arrive with probability $1-(1-\frac{1}{\bigval})^{(\bigval-1)} \rightarrow 1-e^{-1}$, at which point it will be completed with probability $1-(1-\frac{1}{\bigvalvol})^{(1-z)\bigvalvol} \rightarrow 1 - e^{-(1-z)}$. Thus, the optimal policy within this class of policies solves
$$\max_{z \in [0,1]} 1- e^{-z} + (1-e^{-1})(1-e^{-(1-z)}).$$
This single-variable optimization problem has a value of $2 - \frac{1+2\sqrt{e-1}}{e}$, which is achieved at $z^* = 1-\frac{1}{2}\log(e-1)$.

Next we show that the above policy is indeed optimal beyond its subclass and within all policies.
Due to the impacts of submodularity (i.e., because there are diminishing returns to notifying more volunteers about the same task), in order to show that this policy is optimal among all policies, it is sufficient to show that it is not preferable to save any additional volunteer for a future arrival. The marginal value of a volunteer in the first group is equal to $\frac{1}{\bigvalvol}$ times the probability that no other volunteer completes the task, or equivalently, $\frac{1}{\bigvalvol}(e^{-z^*}) = \frac{1}{\bigvalvol}\left(\frac{\sqrt{e-1}}{e} \right)$. The marginal value of saving an additional volunteer in group 1 for a (potential) second task is equal to $\frac{1}{\bigvalvol}$ times the probability that a task arrives within $\bigval$ periods times the probability that no other volunteer completes the task, or equivalently, $\frac{1}{\bigvalvol}(1-e^{-1})(e^{-(1-z^*)}) = \frac{1}{\bigvalvol}\left(\frac{\sqrt{e-1}}{e} \right)$. The marginal value of saving a volunteer from group 1 for a (potential) third task is equal to $\frac{1}{\bigvalvol}$ times the probability that at least two tasks arrive within $\bigval$ periods, or equivalently, $\frac{1}{\bigvalvol}(1-2e^{-1}) < \frac{1}{\bigvalvol}\left(\frac{\sqrt{e-1}}{e} \right)$. 
Thus, we conclude that there is no benefit to saving an additional volunteer from group 1 for any future task. 

Following the same idea, we show there is no benefit from saving any of the volunteers from group 2 for a potential subsequent arriving task.
Conditional on a second task arriving, the marginal value of a volunteer in that second group is $\frac{1}{\bigvalvol}(e^{-(1-z^*)}) = \frac{1}{\bigvalvol}\left(\frac{\sqrt{e-1}}{e-1} \right) > \frac{1}{\bigvalvol}0.762 $. The marginal value of saving a volunteer for a (potential) third task is equal to $\frac{1}{\bigvalvol}$ times the probability that a third task arrives. This probability depends on when the second arrival occurs, but it is at most $1-e^{-1} < 0.762$ (because at most $\bigval$ periods remain). 
Thus we can conclude that there is no benefit to saving a volunteer from group 2 for a third task. 

Putting all this together,
we have shown that the policy ${\pi}^*$ is optimal, where ${\pi}^*$ notifies a fraction $z^* = 1-\frac{1}{2}\log(e-1)$ of volunteers at time $1$ and notifies the remaining fraction if and when another task arrives.
\Halmos
\end{proof}

We now compute the total expected number of completed tasks when following ${\pi}^*$ at each $\tau_i$. Between each $\tau_i$ and $\tau_{i+1}$, the policy completes 0.668 tasks in expectation, as established in the prior claim. By definition, no tasks arrive between $\tau_i + \bigval$ and $\tau_{i+1}$ (recall that $\tau_{i+1}$ is the first arrival after $\tau_i + \bigval$). Therefore, the total expected number of tasks completed by $\pi^*$ is $0.668$ times the expected number of $\tau_i$'s which occurr between times 1 and $T = n^2$. The time between $\tau_i$ and $\tau_{i+1}$ is an independently drawn random variable with an expected value of $2\bigval$ and variance $\bigval^2-\bigval$ (more precisely, it is the sum of a deterministic random variable of length $\bigval$ plus a geometric random variable with parameter $\frac{1}{\bigval}$). Therefore, in the limit as $\bigval \rightarrow \infty$, the expected number of $\tau_i$'s converges to $\frac{\bigval}{2}$.

We have shown that repeating ${\pi}^*$, which is an upper bound on the optimal online policy, completes $0.334 \bigval$ tasks in expectation. Further, since $\mathbf{LP} \geq \bigval$, we have established an upper bound of $0.334$ on the competitive ratio 
for instance $\set{I}_3$, thus completing the proof of Lemma \ref{lem:exq0}. \Halmos
}

\section{Examples Comparing Ex Ante Candidate Solutions}
\label{ex:xstarcandidates}
\begin{figure}[t]
\centering
 \includegraphics[trim=120 550 120 72, clip, width=.75\textwidth]{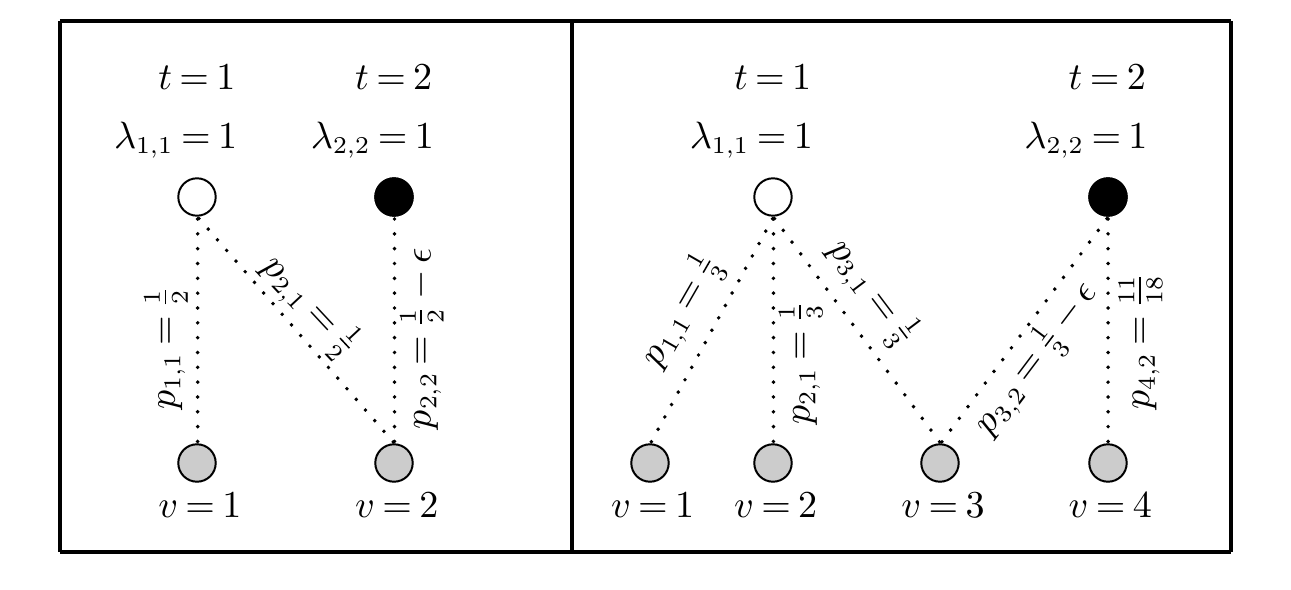}
\caption{\emph{Left:} Visualization of instance $\set{I}_5$. \emph{Right:} Visualization of instance $\set{I}_6$.} \label{fig:exappendix}
\end{figure}
 In the following instances, we first present evidence that $f(\mathbf{x^*_{LP}})$ can be significantly less than $f(\mathbf{x^*})$, where $f(\cdot)$ is defined in \eqref{eq:fdef1}, $\mathbf{x^*_{LP}}$ is the solution to (LP), and $\mathbf{x^*}$ is defined in \eqref{eq:xstar}. We then present evidence that $f(\mathbf{x^*_{LP}})$ can be strictly greater than $f(\mathbf{x^*_{AA}})$ and $f(\mathbf{x^*_{SQ}})$, thus necessitating its inclusion as a candidate.
 \medskip

{\bf Instance $\set{I}_5$:} Suppose $V = 2$, $\D = 2$, $T=2$, and $g(1) = 0$. The arrival probabilities are given by $\lambda_{1, 1} = 1$ and $\lambda_{2, 2} = 1$. The volunteer match probabilities are given by $p_{1, 1} = p_{2,1} = 0.5$ and $p_{2,2} = 0.5-\epsilon$, where $\epsilon << 1$. The left panel of Figure \ref{fig:exappendix} visualizes instance $\set{I}_5$.

\noindent In instance $\set{I}_5$, $f(\mathbf{x^*_{LP}})$ can be significantly less than $f(\mathbf{x^*_{AA}})$ and $f(\mathbf{x^*_{SQ}})$ To see why, note that:
\begin{itemize}
    \item The solution $\mathbf{x^*_{LP}}$ is $x_{1,1,1} = x_{2,1,1} = 1$ and $x_{2, 2, 2} = 0$. Consequently, $f(\mathbf{x^*_{LP}}) = 0.75$.
    \item Suppose the number of steps used when computing $\mathbf{x^*_{AA}}$ is $\FWstep=2$. Then, the solution is $x_{1,1,1} = 1$, $x_{2,1,1} = 0.5$, and $x_{2, 2, 2} = 0.5$, which yields $f(\mathbf{x^*_{AA}}) = 0.875 - 0.5\epsilon$.
    \item The solution $\mathbf{x^*_{SQ}}$ is $x_{1,1,1} = 1$, $x_{2,1,1} = 0$, and $x_{2, 2, 2} = 1$ which attains a value $f(\mathbf{x^*_{SQ}}) = 1 - \epsilon$.
\end{itemize}
\medskip

\noindent With the above solutions,
if $1 >> \epsilon$, $f(\mathbf{x^*_{SQ}}) \approx 1$ and $f(\mathbf{x^*_{AA}}) \approx 0.875$, both of
which represent a significant improvement over $f(\mathbf{x^*_{LP}}) = 0.75$.\footnote{{We remark that if the volunteers' indices were switched, then we would have $f(\mathbf{x^*_{AA}}) > f(\mathbf{x^*_{SQ}})$, which justifies the inclusion of both of these solutions as ex ante candidates.}}

\medskip




{\bf Instance $\set{I}_6$:} Suppose $V = 4$, $\D = 2$, $T=2$, and $g(1) = 0$. The arrival probabilities are given by $\lambda_{1, 1} = 1$ and $\lambda_{2, 2} = 1$. The volunteer match probabilities are given by $p_{1,1} = p_{2,1} = \frac{1}{3}$, $p_{4, 2} = \frac{11}{18}$, $p_{3,1} = \frac{1}{3}$, and $p_{3,2} = \frac{1}{3}- \frac{1}{1000}$. The right panel of Figure \ref{fig:exappendix} visualizes instance $\set{I}_6$.

\noindent In instance $\set{I}_6$, $f(\mathbf{x^*_{LP}})$ is strictly greater than $f(\mathbf{x^*_{AA}})$ and $f(\mathbf{x^*_{SQ}})$. To see this, note that:
\begin{itemize}
    \item The solution $\mathbf{x^*_{LP}}$ is $x_{1,1,1} = 1, x_{2,1,1} = 1, x_{3,1,1} = 1$, and $x_{4,2,2} = 1$, which attains a value $f(\mathbf{x^*_{LP}}) = 1.315$. 
    \item Suppose the number of steps used when computing $\mathbf{x^*_{AA}}$ is $\FWstep=5$. Then, the solution is $x_{1,1,1} = 1, x_{2,1,1} = 1, x_{3,1,1} = 0.6$, $x_{3,2,2} = 0.4$, and $x_{4,2,2} = 1$, which yields $f(\mathbf{x^*_{AA}}) = 1.307$.\footnote{We remark that for any $n > 2$, $f(\mathbf{x^*_{AA}}) < f(\mathbf{x^*_{LP}}).$}
    \item The solution $\mathbf{x^*_{SQ}}$ is $x_{1,1,1} = 1, x_{2,1,1} = 1, x_{3,2,2} = 1$, and $x_{4,2,2} = 1$ which attains a value $f(\mathbf{x^*_{SQ}}) = 1.296$. 
\end{itemize}
\medskip 

\noindent Thus, in instance $\set{I}_6$, $\mathbf{x}^*_{LP} = \text{argmax}_{\mathbf{x} \in \{\mathbf{x^*_{LP}}, \mathbf{x^*_{AA}}, \mathbf{x^*_{SQ}} \}} f(\mathbf{x})$.

\section{Scaled-Down Notification Policy}
\label{app:sdnp}
{\revcolor In this section, we present a second policy for the online volunteer notification problem, which we call our scaled-down notification (SDN) policy. We begin by describing the design of the SDN policy. Like the SN policy, this non-adaptive, randomized policy independently notifies volunteers based on an adjustment of the ex ante solution $\mathbf{x^*}$. We then show that the SDN and the SN policies achieve the same competitive ratio. However, we conclude this section by highlighting that the SDN policy often achieves significantly worse numerical performance the SN policy, hence our focus on the latter.}

In this section, we present our scaled-down notification (SDN) policy which is a non-adaptive randomized policy that independently notifies volunteers according to a  predetermined set of probabilities based on $\mathbf{x^*}$.
The policy relies on the following ideas: (i) 
{\revcolor suppose we can compute the \emph{ex ante} probability that any volunteer $v$ is active at time $t$ when following the SDN policy.}
Let us denote such an ex ante probability by $\beta_{v,t}$. Then if $\don$ arrives at time $t$, we notify $v$ with probability $\cvar x^*_{v,\don,t}/\beta_{v,t}$ {where $\cvar \in [0,\beta_{v,t}/x^*_{v,\don,t}]$.} As a result, she will be
active \emph{and} notified with probability $\cvar x^*_{v,\don,t}$. (ii) If she was the only notified volunteer, then her probability of completing this task would be simply $\cvar x^*_{v,\don,t}p_{v,\don}$.
Even though this is not the case, using the index-based priority scheme and the contribution decoupling idea in Lemma \ref{lem:falg}, we can show her contribution will be proportional to $\cvar x^*_{v,\don,t}p_{v,\don}$. 
(iii) Consequently, we would like to set $\cvar$ as large as possible. 
However, $\cvar$ cannot be larger than $\frac{\beta_{v,t}}{x^*_{v, \don, t}}$ since notification probabilities cannot exceed $1$.
Thus in the design of the policy, we find the largest feasible $\cvar$, which we prove to be $1/(2-q)$ where $q$ is the MDHR of the inter-activity time distribution (see Definition \ref{def:hazard}).

The formal definition of our policy is presented in Algorithm \ref{alg:one}. We now analyze the competitive ratio of the SDN policy. Our main result is the following theorem:

\begin{algorithm}
	\textbf{Offline Phase}:
	    \begin{enumerate}
	        \item Compute $\mathbf{x^*}$ according to \eqref{eq:xstar}.
	    \item Set $\beta_{v,1} = 1$ and $\beta_{v, t} = 1 - \sum_{t'=1}^{t-1} \sum_{\don=1}^{\D} \lambda_{\don,t'} \frac{x^*_{v,\don,t'}}{2-q}(1-G(t - t'))$ for all $v \in [V], t \in [T] \setminus [1]$
	\end{enumerate}
	
	\textbf{Online Phase}:
	\begin{enumerate}
	\item \textbf{For} $t \in [T]$:
	\begin{enumerate}
	\item If a task of type $\don$ arrives in time $t$, then:
	\begin{enumerate}
	    \item {\revcolor \textbf{For} $v \in [V]$:}
	\begin{itemize}
	\item Notify $v$ with probability $\frac{x^*_{v,\don,t}}{(2-q)\beta_{v,t}}$
	\end{itemize}
	\end{enumerate}
	\end{enumerate}
	\end{enumerate}
	\caption{Scaled-Down Notification (SDN) Policy}
	\label{alg:one}
\end{algorithm}

\begin{theorem}[Competitive Ratio of the Scaled-Down Notification Policy] \label{thm:alg1}
Suppose the MDHR of the inter-activity time distribution is $q$. Then the scaled-down notification policy, defined in Algorithm \ref{alg:one},
is $\frac{1}{2-q}(1-\frac{1}{e})$-competitive.
\end{theorem}

We remark that Theorem \ref{thm:alg1} implies that the competitive ratio of our policy improves as $q$ increases. However, a larger value of $q$ does not imply that the probability of notification is uniformly larger. If $q$ increases, the ex ante solution as well as the ex ante probability of being active will also change, both of which affect the notification probability.

The proof of Theorem \ref{thm:alg1} builds on  the ideas described above and consists of several steps. First, in the following lemma, we prove that for any $v \in [V]$ and $t \in [T]$, $\beta_{v,t}$ defined in Algorithm \ref{alg:one} is indeed the probability that $v$ is active at time $t$ under the SDN policy and $\beta_{v,t}$ is at least $\frac{1}{2-q}$.\footnote{We also highlight that computing  $\beta_{v,t}$ for all $v$ and $t$ can be done in polynomial time.}

\begin{lemma}[Volunteer's Active State Probability] \label{lem:beta}
For the SDN policy defined in Algorithm \ref{alg:one}, let $\mathcal{E}_{v,t}$ 
represent the event that volunteer $v$ is active in period $t$. Then for all $v \in [V]$ and all $t \in [T]$, $\P{\mathcal{E}_{v,t}} = \beta_{v,t}$. Further, $\beta_{v,t}\geq \frac{1}{2-q}$.
\end{lemma}


\begin{proof}{Proof:}
We begin by proving that $\beta_{v,t} \geq \frac{1}{2-q}$ for all $v \in [V]$ and $t \in [T]$. Since we set $\beta_{v,1} =1$, this clearly holds for all $v \in [V]$ when $t=1$. Without loss of generality, we will now focus on a particular $v \in [V]$ and $t \in [T]\setminus [1]$. 
 
 Starting from the definition of $\beta_{v,t}$ in Algorithm \ref{alg:one}, we have:
 \begin{align}
    \beta_{v,t} &=  1 - \sum_{t'=1}^{t-1} \sum_{\don=1}^{\D} \lambda_{\don,t'} \frac{x^*_{v,\don, t'}}{2-q}(1-G(t - t')) \nonumber \\
    &=1-\frac{1}{2-q}\left(\sum_{t'=1}^{t-1} \sum_{\don=1}^{\D} \lambda_{\don,t'} x^*_{v,\don, t'}(1-G(t - t'-1)-g(t-t'))\right) \label{eq:betaline0} \\
    &=1-\frac{1}{2-q}\left(\sum_{t'=1}^{t-1} \sum_{\don=1}^{\D} \lambda_{\don,t'} x^*_{v,\don, t'}(1-G(t - t'-1))(1-\frac{g(t-t')}{1-G(t-t'-1)})\right)  \label{eq:betaline0.5} \\
    &\geq 1-\frac{1}{2-q}\left(\sum_{t'=1}^{t-1} \sum_{\don=1}^{\D} \lambda_{\don,t'} x^*_{v,\don, t'}(1-G(t - t'-1))(1-q)\right) \label{eq:betaline1} \\
    &\geq 1-\frac{1-q}{2-q} \label{eq:betaline2} \\
    &=\frac{1}{2-q} \nonumber
\end{align}
Line \eqref{eq:betaline0} simply follows by observing that $G(t - t') = G(t - t' -1) + g(t-t')$. In the next line, we factor out $(1-G(t - t'-1))$.\footnote{{If $1-G(t - t'-1) = 0$, then we must also have $g(t-t') = 0$. Thus, in Line \eqref{eq:betaline0.5}, we preserve the equality by following our convention that if the fraction is $\frac{0}{0}$, we define it to be equal to 1.}}
Line \eqref{eq:betaline1} comes from applying the definition of the MDHR. Because $\mathbf{x^*} \in \set{P}$ (see Definition \ref{def:P}), in line \eqref{eq:betaline2} we apply the bound given by constraint \eqref{eq:lpcon2} for $t-1$, namely, $\sum_{t'=1}^{t-1} \sum_{\don=1}^{\D} \lambda_{\don,t'} x^*_{v,\don, t'}(1-G(t - t'-1)) \leq 1$. This holds for any $v \in [V]$ and any $t \in [T]\setminus [1]$, which implies that $\beta_{v,t} \geq \frac{1}{2-q}$ for all $v \in [V]$ and $t \in [T]$.

 We now use total induction to prove that $\P{\mathcal{E}_{v,t}} = \beta_{v,t}$. For notation, we will use $\mathcal{E}^c$ to refer to the complement of event $\mathcal{E}$. At $t=1$, we have $\P{\mathcal{E}_{v,t}}=1$ and $\beta_{v,t} = 1$ by definition. Now we assume that $\beta_{v,t} = \P{\mathcal{E}_{v,t}}$ for all $t \in [k]$. 
 We prove the claim for $t = k+1$ by {\revcolor showing $\P{\mathcal{E}_{v,k+1}^c} = 1-\beta_{v, k+1}$.}
 To be inactive at time $k+1$, a volunteer must be active in some prior period $t' \in [k]$, must be notified in $t'$, and must not become active again by time $k+1$. Thus, to compute the probability that a volunteer is inactive at time $k+1$, we can sum the probabilities of $k$ disjoint events:
 \begin{equation}
   \P{\mathcal{E}_{v,k+1}^c} = \sum_{t'=1}^{k} \sum_{\don=1}^{\D} \lambda_{\don,t'} \P{\mathcal{E}_{v,t'}}\frac{x^*_{v,\don, t'}}{(2-q)\beta_{v,t'}}(1-G(k+1-t')) \nonumber
 \end{equation}
 Plugging in the inductive hypothesis that $\P{\mathcal{E}_{v,t'}} = \beta_{v,t'}$ for $t' \in [k]$, we see that these terms cancel, leaving us with 
 \begin{equation}
     \P{\mathcal{E}_{v,k+1}^c} = \sum_{t'=1}^{k} \sum_{\don=1}^{\D} \lambda_{\don,t'} \frac{x^*_{v,\don, t'}}{2-q}(1-G(k+1-t')) \LinesNotNumbered
 \end{equation}
 Noting that this sum is definitionally equivalent to $1-\beta_{v, k+1}$ completes the proof. \Halmos 
 \end{proof}

Next, utilizing the index-based priority scheme (in Definition \ref{def:priority}) and the contribution decoupling idea {(in Lemma \ref{lem:falg})}, we lower bound the contribution of each volunteer according to their priority in the following lemma:

\begin{lemma}[Volunteer Priority-Based Contribution under the SDN Policy]
\label{lem:contributionsalg1}
Under the index-based priority scheme (in Definition \ref{def:priority}) and the SDN policy, for any $\mathbf{x} \in \set{P}$, the contribution of volunteer $v \in [V]$, i.e., the expected number of tasks she completes, is at least $\frac{1}{2-q} f_v(\mathbf{x})$, with $f_v(\cdot)$ defined in \eqref{eq:decouple}.
\end{lemma} 
\begin{proof}{Proof:}
First, we focus on a particular arrival $\don \in [\D]$ at a particular time $t \in [T]$ and we show that volunteer $v \in [V]$ completes the task with probability at least
 $\frac{1}{2-q}\left(\prod_{u < v} (1-x^*_{u, \don, t}p_{u,\don})\right)x^*_{v, \don, t} p_{v,\don}$. Then we use linearity of expectations to finish the proof. 
{\revcolor Under an index-based priority scheme, a volunteer $v \in [V]$ completes a task if (i) she responds and (ii) no lower-indexed volunteer responds. These events may not be independent\footnote{Recall that a volunteer responds if she is active, notified, and matches with the task. There can be correlation between volunteers' states based on the past sequence of arrivals, which means that volunteer responses can be correlated.}; however, for event (ii) to occur, it is sufficient (but not necessary) that all lower-indexed volunteers are either not notified or do not match. 
Because notifications and matching are independent across all volunteers, this observation allows us to lower bound the probability that $v$ completes the task.
\begin{align*}
    \P{v \text{ completes the task}}&=\P{v \text{ responds and no volunteer with lower index responds}} \\ &\geq \P{v \text{ responds}}\P{ \cap_{u < v} \{u \text{ not notified or does not match}\}} \\
    & = \frac{x^*_{v, \don, t}}{(2-q)} p_{v,\don}\prod_{u < v} (1-x^*_{u, \don, t}p_{u,\don})
\end{align*}
This proves that for any arrival $\don \in [\D]$ at time $t \in [T]$, volunteer $v \in [V]$ completes the task with probability at least
 $\frac{1}{2-q}\left(\prod_{u < v} (1-x^*_{u, \don, t}p_{u,\don})\right)x^*_{v, \don, t} p_{v,\don}$. 
 
 Using linearity of expectations, we compute the expected number of tasks completed by $v$ as follows
$$ \frac{1}{2-q} \sum_{t=1}^T \sum_{\don =1}^\D \lambda_{\don,t} \left(\prod_{u < v} (1-x^*_{u, \don, t}p_{u,\don})\right)x^*_{v, \don, t} p_{v,\don} = \frac{1}{2-q}f_v(\mathbf{x^*}).$$
This completes the proof of Lemma \ref{lem:contributionsalg1}.} \Halmos \end{proof}

To finish the proof of Theorem \ref{thm:alg1}, note that Lemma \ref{lem:contributionsalg1} implies that each volunteer completes at least $\frac{1}{2-q}f_v(\mathbf{x^*})$ tasks in expectation. By linearity of expectations and Lemma \ref{lem:falg}, the expected total number of tasks completed by volunteers must be at least $\frac{1}{2-q}f(\mathbf{x^*})$. Since $f(\mathbf{x^*}) \geq (1-\frac{1}{e})\mathbf{LP}_\set{I}$ (see Proposition \ref{prop:fxstar}), it immediately follows that the SDN policy is  $\frac{1}{2-q}(1-\frac{1}{e})$-competitive. \Halmos

\begin{figure}[t]
\centering
\scalebox{1.0}[1.0]{
 \includegraphics[width=.98\textwidth]{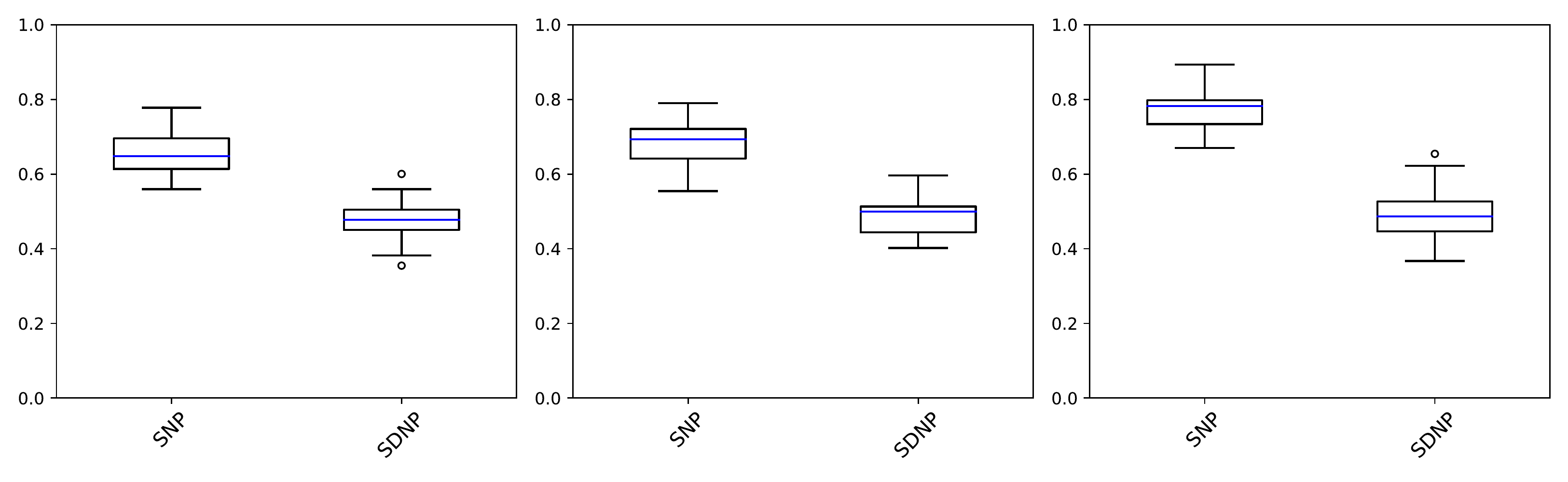}}
\caption{The fraction of $\mathbf{LP}$ achieved in Locations (a), (b), and (c) (left, middle, and right, respectively) by the SN and SDN policies assuming a deterministic inter-activity time.}
\label{fig:numericssdn}
\end{figure}

The competitive ratio of the SDN policy is identical to the competitive ratio of the SN policy, implying that in the worst case they guarantee the same performance of $\frac{1}{2-q}(1-\frac{1}{e})\mathbf{LP}$.  In addition, the SDN policy has nearly identical computational complexity as the SN policy, which we formalize in the following remark. 

\begin{remark}[Computational Complexity of SDN and SN Policies]
\label{remark:compcomplex}
First, we note that because both policies are non-adaptive, all of the computation can be done in advance. 
Next, we claim that the computational complexities of our policies are nearly identical: to see this, observe both policies require computing the same ex ante solution. 
Given an ex ante solution, the steps in the SN policy (Algorithm \ref{alg:two}) have complexity of $O(V\D T^2 + V^2\D T)$, while the steps in the SDN policy (Algorithm \ref{alg:one}) have complexity of $O(V\D T^2)$.\footnote{For two functions $d, l : \mathbb{N}\rightarrow \mathbb{R}$, $l(n) = O(d(n))$ if $\lim \sup_{n \rightarrow \infty} \frac{|l(n)|}{d(n)} < \infty $.} Finally, we remark that computing the ex ante solution according to \eqref{eq:xstar}, which is a common step shared by both policies, requires comparing three candidates: $\mathbf{x^*_{LP}}, \mathbf{x^*_{AA}},$ and $\mathbf{x^*_{SQ}}$. These ex ante candidates require solving $1$, $n$, and $V$ linear programs, respectively, where $\frac{1}{n}$ is the step size of the Frank-Wolfe variant used to compute $\mathbf{x^*_{AA}}$. These linear programs consist of $V\D T$, $V \D T$, and $\D T$ variables, respectively, and $T(\D + V+V\D )$, $T(V+V\D )$, and $T(1 + \D)$ constraints.
\end{remark}

Despite the identical competitive ratio and nearly identical computational complexity, the SN policy generally proves numerically superior to the SDN policy, which we demonstrate in both the FRUS setting and in another numerical example. In Figure \ref{fig:numericssdn} we compare the performance of the two policies in the FRUS setting described in Section \ref{sec:data} (i.e., the same setting shown in Figure \ref{fig:numerics1}). Both policies significantly outperform their theoretical guarantee, in part because using $\mathbf{x^*}$ as an ex ante solution as defined in \eqref{eq:xstar} improves upon using $\mathbf{x^*_{LP}}$ by an average of $12\%$ across the largest FRUS locations, including Locations (a), (b), and (c) shown in Figure \ref{fig:numericssdn}. However, the SN policy performs significantly better than the SDN policy across all locations.

Furthermore, instance $\set{I}_4$ (defined in the proof of Proposition \ref{prop:ubexante} in Section \ref{sec:hardness} and visualized in the bottom-right panel of Figure \ref{fig:exhardness}) provides an example where the SDN policy performs only half as well as the SN policy when $q << 1$. To show this, we first demonstrate that under the SDN policy, the volunteer is active and notified with probability $\frac{1}{2-q}$ in both periods. {This implies that she completes $\frac{1}{2-q}f_1(\mathbf{x^*})$ tasks in expectation, which is exactly the lower bound established in Lemma \ref{lem:contributionsalg1}.

Because $x^*_{1,1,1} = 1$ and $\beta_{1,1} =1$, the SDN policy will notify volunteer $1$ at time $1$ with probability $\frac{1}{2-q}$. Because $x^*_{1,2,2} = 1$, when a task of type $2$ arrives at time $2$, the SDN policy will notify her with probability $\frac{1}{(2-q)\beta_{1,2}}$. According to Lemma \ref{lem:beta}, she will be active in period $2$ independently with probability $\beta_{1,2}$.} Thus, the SDN policy achieves a value of $\frac{\epsilon}{2-q} + q\beta_{1,2}\frac{1}{(2-q)\beta_{1,2}} = \frac{1}{2-q} \mathbf{LP}_{\set{I}_4}$. 

On the other hand, we show that the SN policy performs nearly twice as well by not notifying the volunteer in period 1 and instead saving her for period 2. (i.e., $\y_{1,1,1} = 0$ and $\y_{1,2,2} = 1$)
To see this, note that the SN policy solves a DP starting from $J_{1,3} = 0$. Working backwards, the DP solution for period $2$ is $\y_{1,2,2} = 1$ and $J_{1,2} = q$. To evaluate the DP solution for period 1, we note that $r_{1,1,1} + q^2 = \epsilon + q^2 < \J_{1,2}$, assuming $q<1$ and $\epsilon$ is chosen to be sufficiently small. Thus, $\y_{1,1,1} = 0$, so the SN policy does not notify the volunteer in period 1. Instead, she is saved for period 2, where she completes a task with probability $q$. Thus, the expected number of completed tasks under the SN policy is $\frac{q}{q+\epsilon}\mathbf{LP}_{\set{I}_4}$. Comparing this to the SDN policy, which completes $\frac{1}{2-q} \mathbf{LP}_{\set{I}_4}$ tasks, we observe that {for $\epsilon << q << 1$,} we have $\frac{q}{q+\epsilon} \approx 2 (\frac{1}{2-q} )$.

Intuitively, this gap in numerical performance exists because the design of the SN policy makes individualized decisions explicitly intended to {optimally} resolve the trade-off between notifying a volunteer now or keeping her active for later {based on $\mathbf{x^*}$}. On the other hand, the design of the SDN policy aims to uniformly scale down $\mathbf{x^*}$.
As a result, the SDN policy's numerical performance is not substantially better than its worst-case guarantee. {\revcolor We remark that the SDN policy is designed based on a ``scale-down'' factor of $c=\frac{1}{2-q}$, which achieves the best worst-case guarantee. However, given a specific problem instance, it is natural to consider optimizing the scale factor $c$ to achieve the best expected numerical performance when volunteers are notified with probability $x^*_{v,\don,t} \times \min\{1, \cvar/\beta_{v,t}\}$.  In the example just considered (instance $\set{I}_4$),  the best scale factor is exactly $c=\frac{1}{2-q}$, which demonstrates that this optimization may not improve performance.}

\section{Additional Numerical Analysis}
\label{app:robustness}
\begin{table}[t]
    \caption{Mean percent change in the number of completed tasks in the presence of estimation errors.}
    \label{table:robustness}
    \centering
    \begin{tabular}{c r r r r r r}
    \toprule
    & \multicolumn{3}{c}{SDN Policy} & \multicolumn{3}{c}{SN Policy} \\
\cmidrule(lr){2-4}\cmidrule(lr){5-7}
    & Loc. (a) & Loc. (b) & Loc. (c) & Loc. (a) & Loc. (b) & Loc. (c) \\
    \cmidrule(lr){2-4}\cmidrule(lr){5-7}
         $\hat{p}_{v,\don} \in  [0.9p_{v,\don}, 1.1p_{v, \don}]$  
         &$-0.13\%$&$0.59\%$&$0.84\%$&$0.33\%$&$0.30\%$&$-0.48\%$ \\
$\hat{\lambda}_{\don,t} \in [0.9\lambda_{\don,t}, 1.1\lambda_{\don,t}]$ &$0.05\%$&$0.27\%$&$-0.45\%$&$0.01\%$&$-0.71\%$&$-1.14\%$ \\
\bottomrule
    \end{tabular}

\end{table}
In this section, we demonstrate that our two policies are robust in the presence of estimation error. 
In Table \ref{table:robustness}, we report {the average percent change in the number of tasks completed in settings where we misestimate the actual arrival rates or match probabilities.} 
Specifically, in the first row, we consider cases when $\hat{p}_{v,\don}$ is uniformly distributed in the interval $[0.9p_{v,\don}, 1.1p_{v, \don}]$. In the second row, we consider cases when $\hat{\lambda}_{\don,t}$ is uniformly distributed in the interval $[0.9\lambda_{\don,t}, 1.1\lambda_{\don,t}]$. {We generate 10 different perturbations of each of the two model primitives, and for each perturbation, we simulate our policies 500 times. } In all cases, we follow the current practice at most FRUS locations and assume the inter-activity time is deterministically equal to seven days (as discussed further in Section \ref{sec:data}). 

As we observe in Table \ref{table:robustness}, the decrease in the number of completed tasks is never more than 1.14\% for either policy. In fact, in some settings (e.g. in Location (b) for the SDN policy), the expected number of tasks completed actually increases. These results suggest that even if the ex ante solution is computed based on misestimated model primitives, notifying volunteers based on such a solution (after modifications according to  our two policies) results in good performance.

\end{APPENDICES}

\end{document}